\documentclass[format=acmsmall,authorversion=true]{acmart}

\usepackage{acm-ec-26}
\usepackage{booktabs} 

\AtBeginDocument{%
  \fancypagestyle{firstpagestyle}{%
    \fancyhf{}%
  }%
}

\setcitestyle{authoryear}

\usepackage{latexsym,amsfonts}
\usepackage{amsmath}
\usepackage{array}

\usepackage{graphicx,color}
\usepackage{bbm}

\usepackage{algorithmic}
\usepackage{algorithm}
\usepackage{tikz}
\usepackage{url}

\usepackage{booktabs}
\usepackage[table]{xcolor}
\usepackage{makecell}
\usepackage{wrapfig}

\usepackage{amsthm}
\AtEndPreamble{%
\newtheorem{claim}[theorem]{Claim}
\theoremstyle{acmdefinition}

}

\usepackage{thm-restate}

\definecolor{blue}{RGB}{0,82,147} 
\definecolor{red}{RGB}{202,033,063}
\definecolor{green}{RGB}{98,158,31} 

\colorlet{green}{green!50!black}
\usepackage[capitalize]{cleveref}
\crefname{claim}{Claim}{Claims}

\newtheorem{new-claim}{Claim}

\newcommand{\X}{\ensuremath{\mathcal{X}}}
\newcommand{\Y}{\ensuremath{\mathcal{Y}}}
\newcommand{\M}{\ensuremath{\mathcal{M}}}
\newcommand{\N}{\ensuremath{\mathcal{N}}}
\newcommand{\I}{\ensuremath{\mathcal{I}}}

\newcommand{\Iconstr}{$\mathcal{I}\!$-constrained}

\author{Telikepalli Kavitha}
\affiliation{%
\institution{Tata Institute of Fundamental Research}
  \city{Mumbai}
  \country{India}}
  
\author{Jannik Matuschke}
\affiliation{%
\institution{KU Leuven}
  \city{Leuven}
  \country{Belgium}}
  
\author{Ulrike Schmidt-Kraepelin}
\affiliation{%
  \institution{TU Eindhoven}
  \city{Eindhoven}
  \country{The Netherlands}}
  
\author{} 

\title{Condorcet Dimension and Pareto Optimality\\ for Matchings and Beyond}

\begin{abstract}
We study matching problems in which agents form one side of a bipartite graph and have preferences over objects on the other side.
A central solution concept in this setting is popularity: a matching is \emph{popular} if it is a \emph{(weak) Condorcet winner},\footnote{For the sake of this abstract we do not distinguish between weak Condorcet winners and Condorcet winners.} meaning that no other matching is preferred by a strict majority of agents. 
It is well known, however, that Condorcet winners need not exist. We therefore turn to a natural and prominent relaxation. A set of matchings is a \emph{Condorcet-winning set} if, for every competing matching, a majority of agents prefers their favorite matching in the set over the competitor. The 
\emph{Condorcet dimension} is the smallest cardinality of a Condorcet-winning set. For matchings, this notion admits a natural interpretation as a variant of resource augmentation: it captures how much we must relax item capacities in order to guarantee the existence of a popular outcome.

Our main results reveal a connection between Condorcet-winning sets and Pareto optimality. We show that any Pareto-optimal set of two matchings is, in particular, a Condorcet-winning set. This implication continues to hold when we impose matroid constraints on the set of matched objects, and even when agents’ valuations are given as partial orders. The existence picture, however, changes sharply with partial orders. While for weak orders a Pareto-optimal set of two matchings always exists, this is—surprisingly—not the case under partial orders. Consequently, although the Condorcet dimension for matchings is $2$ under weak orders (even under matroid constraints), this guarantee fails for partial orders: we prove that the Condorcet dimension is $\Theta(\sqrt{n})$, and rises further to $\Theta(n)$ when matroid constraints are added. On the computational side, we show that, under partial orders, deciding whether there exists a Condorcet-winning set of a given fixed size is NP-hard.
The same holds for deciding the existence of a Pareto-optimal matching, which we believe to be of independent interest. Finally, we also show that the Condorcet dimension for a related problem on arborescences is also $2$. 
\end{abstract}

\begin{document}

\begin{titlepage}

\maketitle
\setcounter{tocdepth}{1} 
\tableofcontents
    
\end{titlepage}

\newpage

\section{Introduction} 
\label{sec:intro}

We study the fundamental problem of matching under one-sided preferences. Such problems arise naturally in allocation settings in which agents express preferences over a set of (typically scarce) resources, while the other side is passive and subject only to capacity or other feasibility constraints. Prominent examples include assigning students to seats in university courses, applicants to public housing units, or computational tasks to machines. Formally, our input is a bipartite graph $G = (A \cup O,E)$ where nodes in $A$ are {\em agents} and nodes in $O$ are {\em objects}, and every agent $a \in A$ submits a preference relation $\succ_a$ over the set of objects they are adjacent to in the graph $G$. 

A prominent solution concept in this setting is \emph{popularity} \cite{gardenfors1975match,AIK+07b}, the analog of a weak Condorcet winner in voting. Condorcet winners are defined via pairwise comparisons between alternatives (here, matchings), where $\varphi(M,N)$ is the number of agents who strictly prefer the alternative $M$ over $N$. In the matching context, we assume agents compare matchings by comparing the objects they are assigned under $M$ and $N$. 
An alternative $M$ is a (weak) Condorcet winner if, for every other alternative $N$, $\varphi(M,N)$ is (weakly) higher than $\varphi(N,M)$.

As in classical social choice theory, however, weak Condorcet winners need not exist. While much of social choice theory studies how to select a single alternative in such scenarios, the idea of selecting a \emph{Condorcet-winning set} \citep{elkind2015condorcet} has recently attracted renewed attention: 
A set of alternatives $\M$ is a \emph{(weak) Condorcet-winning set}, if, for every competing alternative $N$, the number of agents that strictly prefer their favorite alternative from $\M$ to the competitor $N$, denoted by $\varphi(\M,N)$, is (weakly) larger than $\varphi(N,\M)$, which is defined analogously. The (weak) Condorcet dimension of an instance is the minimum size of a (weak) Condorcet-winning set. \citet{CLR+25a} recently resolved the long-standing open question whether the Condorcet dimension of an election can be bounded by a constant by showing an upper bound of $6$, which was subsequently improved to $5$ \cite{song2026few}. In this work, we study the Condorcet dimension of matching-induced elections. Crucially, the results discussed above assume strict rankings over alternatives—an assumption that fails in our case, since agents are indifferent between matchings whenever they receive the same object. Hence, none of the above results apply to our setting.

While Condorcet-winning sets are typically viewed as permitting multiple alternatives (here: matchings) to be selected as winners, in our setting they admit a second, natural interpretation in terms of \emph{resource augmentation}. If the object side of the bipartite graph consists of resources that can be replicated—e.g., by adding seats in a course or acquiring additional machines—then selecting a weakly Condorcet-winning set can be seen as selecting a single matching in a capacity-augmented graph that defeats every matching in the original graph in pairwise majority comparisons. In particular, our results will imply that doubling the capacity of every object suffices to guarantee the existence of a weak Condorcet winner in any instance.

Our work begins with the classic one-sided matching model with strict rankings, but quickly moves beyond this setting to capture a broader range of applications. For instance, in course allocation, an instructor may face constraints on which time slots can be offered simultaneously. Such restrictions can be modeled as matroid constraints on the set of matched objects, and due to their generality, matroid variants of the popular matching problem have been studied extensively \cite{Kami17a,CKY24a}. Our matroid-constrained matching model has the additional benefit of generalizing other graph structures, such as branchings \cite{KMS+25a}, which have application in the democratic paradigm of \emph{liquid democracy} \citep{SchmidtKraepelin2023PhD}. Finally, we study the Condorcet dimension of arborescences, generalizing branchings in our setting.

Much of the matching literature assumes that agents report strict or weak rankings over the objects they may be assigned. However, it has been repeatedly argued that even weak rankings are a strong behavioral assumption. For example, \emph{semiorders} \citep{luce1956semiorders} have been proposed for settings in which agents have underlying utilities but cannot reliably distinguish objects whose utilities are sufficiently close. In such models, the induced indifference relation (i.e., “neither object is strictly preferred to the other”) need not be transitive, so the resulting preferences do not form a weak order. We generalize further by allowing agents to report arbitrary partial orders.

\bigskip

\subsection{Our Contribution}
Our main result establishes a close connection between Pareto-optimal sets and weak Condorcet-winning sets. Specifically, every Pareto-optimal set of matchings of size at least $2$ is weakly Condorcet-winning, and this remains true under matroid constraints. Under strict rankings over objects, such sets are in fact (strictly) Condorcet-winning in essentially all instances. The only exception occurs when there exists a matching that assigns every agent their top choice; in that case, that single matching is already a Condorcet winner. We first present the unconstrained case under strict rankings as a warm up and then generalize to the matroid setting and partial orders.

\begin{wraptable}{r}{0.44\textwidth}
\vspace{-0.6\baselineskip}
\centering
\setlength{\tabcolsep}{6pt}
\renewcommand{\arraystretch}{1.05}
\small
\begin{tabular}{r|ccc}
 & M & MM & A \\
\midrule
strict ranking  & \cellcolor{gray!15}$2$ & \cellcolor{gray!15}$2$ & \cellcolor{gray!15}$2$ \\
weak ranking   & $2$ & $2$ & $2$ \\
partial order & $\Theta(\sqrt{n})$ & $ n $ & $2$
\end{tabular}

\vspace{10pt}
\caption{Summary of our results for the (weak) Condorcet dimension for matchings (M), matroid-constrained matchings~(MM), and arborescences~(A). Grey cells indicate bounds on the Condorcet dimension, white cells on the weak Condorcet dimension. All bounds are tight, except for a $\sqrt{2}$-factor for (M, partial order).}
\label{tableIntro}
\vspace{-1.5\baselineskip}
\end{wraptable}

Under strict and weak rankings, Pareto-optimal sets of matchings of size $2$ are guaranteed to exist. Combined with the theorem above, this yields a bound of $2$ on the Condorcet dimension for matroid-constrained matchings under strict preferences, and a bound of $2$ for the weak Condorcet dimension when preferences are weak rankings.\footnote{Note that under weak rankings, no non-trivial bounds on the Condorcet dimension are possible: If all agents are indifferent among all objects, then the only Condorcet-winning set is the one including all alternatives. Hence we focus on the weak Condorcet dimension for all settings except when preferences are strict.}
Somewhat surprisingly, this line of reasoning fails for partial orders: Pareto-optimal sets of matchings may not exist,\footnote{We remark that we employ a definition of Pareto optimality that interprets the incomparability relation of a partial order as \emph{indifference}; see \Cref{sec:prelim} for details.} and therefore the theorem does not imply the existence of small weakly Condorcet-winning sets. Indeed, the worst-case weak Condorcet dimension for matchings with partially ordered preferences is $\Theta(\sqrt{n})$. This result is established via a lower-bound instance of dimension larger than $\sqrt{n}$ and an algorithm that constructs a weakly Condorcet-winning set of size at most $\lceil \sqrt{2n}\rceil$. 
When allowing matroid constraints, the weak Condorcet dimension increases further to $n$. 

The fact that partial orders lead to a non-constant weak Condorcet dimension raises computational questions. In particular, deciding whether an instance admits a weakly Condorcet-winning set of a given size turns out to be NP-complete. In addition, deciding whether a Pareto-optimal matching exists is NP-complete as well; we believe this result to be of independent interest.

Finally, when considering arborescences instead of matchings the (weak) Condorcet dimension is again bounded by $2$. In fact, there always exists a set of size $2$ that assigns every agent a maximal element of their partial order. We summarize all of our bounds in \Cref{tableIntro}.

\bigskip

\subsection{Related Work}

The Condorcet dimension of matching problems was studied independently and in parallel by \citet{connor2025popular}. While the projects differ in focus, one result overlaps: Algorithm 1 in their paper implies that the weak Condorcet dimension is bounded by $2$ for one-sided matchings under weak rankings. We obtain the same bound via our link to Pareto-optimality, which yields a range of significantly simpler algorithms for finding Condorcet sets via Pareto-optimal sets: For example, for weak rankings, one call to a maximum-weight bipartite matching algorithm suffices (see \cref{sec:pareto-implies-condorcet-matroids}); for strict rankings, even a simple round robin procedure works (see \cref{sec:warmup}). \citet{connor2025popular} also study one- and two-sided matching problems with weighted agents, but do not consider matroid constraints, arborescences, or partial orders. They also focus on the weak Condorcet dimension; our results additionally imply bounds for the (strict) Condorcet dimension under strict rankings. An interesting open question that remains from their work is whether the weak Condorcet dimension for the roommates problem is $2$ or $3$.

The concept of \emph{Condorcet-winning sets} in ordinal voting was introduced by \citet{elkind2015condorcet}, who show the existence of elections with Condorcet dimension $3$ and prove an upper bound of $\lceil \log_2 m\rceil$ on the minimum size of a Condorcet-winning set, where $m$ is the number of alternatives. Consequently, it remained open whether the worst-case Condorcet dimension is bounded by a constant, until \citet{CLR+25a} proved an upper bound of $6$, which \citet{song2026few} improved to $5$. Whether the final answer is $3,4$ or $5$ remains open. When voters’ preferences are induced by points in $\mathbb{R}^2$, \citet{lassota2024condorcet} showed the worst-case Condorcet dimension to be $2$ or $3$. 

Importantly, the above results assume that each voter submits a \emph{strict ranking over all alternatives}. This assumption is inherently violated in matching-induced elections: even if agents have strict rankings over objects, these induce only weak rankings over matchings, since an agent is indifferent among matchings where they receive the same object. To the best of our knowledge, the work of \citet{connor2025popular} and the present paper are the first to study the Condorcet dimension when preferences among alternatives are not strict. 

Finally, other set-valued relaxations of the Condorcet principle have been studied. For instance, \citet{charikar2026approximately} prove the existence of small \emph{approximately dominating sets}, an approximation of the classical notion of a \emph{dominating set} (where every outside candidate is beaten by some member of the set). Note that every dominating set is in particular a Condorcet-winning set.

\subsubsection*{Popular structures} The concept of popularity was suggested by \citet{gardenfors1975match} in the context of matchings under two-sided preferences, where it relaxes the solution concept of \emph{stability}. \citet{AIK+07b} initiated the study of popularity for matching under one-sided preferences by providing structural insights and a polynomial-time algorithm for finding popular matchings when they exist. Subsequent work extended these methods, for example to the setting of capacitated houses allocation \citep{manlove2006popular} or assignments (perfect matchings) \citep{kavitha2022popular}. 

\citet{Kami17a} initiated the study of popular matchings under matroid constraints, but allows one matroid per object. Closer to our setting is the model of \citet{KMS+25a}, which modifies the classic matching problem in two respects: (i) feasibility is restricted by a matroid constraint on the set of assigned objects, and (ii) attention is restricted to matchings of maximum cardinality. \citet{KMS+25a} demonstrate that this model unifies several previously studied models; as a consequence, their algorithm for finding popular structures does not only solve the popular matching problem but also extends to popularity notions for graph structures in directed graphs, including branchings and arborescences (see below). Our matroid-constrained model adopts aspect~(i) but does not impose requirement~(ii); determining which of our results extend to the maximum-cardinality regime remains an interesting direction for future work. \citet{csaji2025popular} further extend the model of \citet{KMS+25a} to weighted edges. 
For graph structures in directed graphs, the study of popular branchings was initiated by \citet{kavitha2022popularBranching}. The algorithms of \citet{kavitha2022popularBranching} were generalized to a setting with weighted agents \cite{natsui2022finding} and extended from branchings to arborescences \citep{KMS+25a}.

\section{Preliminaries} \label{sec:prelim}

We start with a few general definitions that will be helpful throughout the paper. We then define our three types of combinatorial elections in \Cref{sec:combelec} and the applied solution concepts in \Cref{sec:solutionConcepts}.

\subsubsection*{General concepts and notation}
For an integer $k \in \mathbb{N}$, we write $[k]$ for the set $\{1,\dots,k\}$.
For a digraph $D = (V, E)$, we let $\delta^-_D(v)$ denote the set of arcs in $E$ entering $v \in V$.

A branching is a graph structure that appears in two different contexts throughout the paper (once as an alternative in a combinatorial election, once as a tool for analysis). Given a directed graph $D=(V,E)$, a subset of arcs $B \subseteq E$ is a \emph{branching} if it satisfies two conditions: (i) each node has at most one incoming arc in $B$, i.e., $|B \cap \delta_D^-(v)| \le 1$ for all $v \in V$, and (ii) $B$ contains no cycle. When we do not specify an underlying graph and simply say that $B$ is a branching on a set $V$, we mean that $B$ is a branching in the complete directed graph induced by the node set $V$.

\subsection{Combinatorial Elections} \label{sec:combelec} \label{sec:prelim-matroids}

We say that a preference relation $\succ$ is a \emph{partial order}, if it is irreflexive (not $x \succ x$), asymmetric ($x \succ y$ implies $y\not \succ x$), and transitive ($x \succ y$ and $y \succ z$ implies $x\succ z$). Given a partial order $\succ$, we define the \emph{induced indifference relation} as $x \sim y \Leftrightarrow x \not \succ y \text{ and } y \not \succ x$. A partial order is a \emph{weak ranking} if the indifference relation is transitive. Lastly, a weak ranking $\succ$ is a \emph{strict ranking}, if the relation is complete (for all $x,y$ either $x \succ y$ or $y \succ x$). 

\subsubsection*{Matchings}
We are given a bipartite graph $G=(A \cup O, E)$, where $A$ is the set of \emph{agents} and $O$ is the set of \emph{objects}, and $n = |A|$. Each agent $a$ has a preference relation $\succ_a$ over the objects adjacent to $a$ in $G$. The relation $\succ_a$ may be a strict ranking, a weak ranking, or a partial order; we specify the assumed preference model locally. We denote a \emph{matching instance} by $(G,\succ)$, where $\succ$ represents~$(\succ_a)_{a \in A}$.

A \emph{matching} is a set of edges $M \subseteq E$ such that every node is incident to at most one edge in $M$. A \emph{$k$-matching} is a set of edges $M \subseteq E$ such that every node is incident to at most $k$ edges; it is well known that every $k$-matching can be decomposed into $k$ (ordinary) matchings. An \emph{$A$-perfect matching} is a matching in which every agent in $A$ is incident to exactly one edge. For a matching $M$ and an agent $a \in A$, we write $M(a)$ for the object to which $a$ is matched, with $M(a)=\emptyset$ if $a$ is unmatched.
We also use $M(A') := \{M(a) : a \in A', M(a) \neq \emptyset\}$ for $A' \subseteq A$ to indicate the set of objects matched to the agents in $A'$.
For a $k$-matching $M$, we use the same notation, where $M(a)$ denotes the set of objects matched to $a$.

We lift the preferences of agents over objects to preferences over matchings: For two matchings $M, M'$ agents compare them by their assigned object, with any assigned object being preferred over being unmatched, i.e., $M \succ_a M' \Leftrightarrow M(a) \succ_a M'(a)$, where~$o \succ_a \emptyset$ for all objects $o$ adjacent to $a$.

\subsubsection*{Arborescences}
We are given a directed graph $D = (A \cup \{r\}, E)$, where $A$ corresponds to the agents and $r$ is a (dummy) root node. Every agent $a$ has a partial order $\succ_a$ over their incoming arcs $\delta_D^-(a)$. An $r$-\emph{arborescence} $T$ is a branching in $D$, with the additional constraint that the only node not having an incoming arc in $T$ is the root $r$. An \emph{arborescence instance} can be summarized by $(D, \succ)$. 

We lift the preferences of agents over incoming arcs to preferences over $r$-arborescences. For two $r$-arborescences $T,T'$ we assume that agents compare them by their corresponding incoming arcs. Formally, we define $T \succ_a T' \Leftrightarrow T\cap \delta_D^{-}(a) \succ_a T'\cap \delta_D^{-}(a)$.

\subsubsection*{Matching under matroid constraints}

A \emph{matroid} $(O, \I)$ on a ground set $O$ is a family of subsets $\I \subseteq 2^O$ with the following properties:
\begin{enumerate}
    \item $\emptyset \in \I$
    \item If $X \in \I$ and $Y \subseteq X$ then $Y \in \I$.
    \item If $X, Y \in \I$, then for every $y \in Y \setminus X$ there is $x \in X$ such that $X \setminus \{x\} \cup \{y\} \in \I$.
\end{enumerate}
The sets in $\I$ are called \emph{independent}.
Inclusionwise maximal independent sets are called \emph{bases}.
It is well-known that all bases of a given matroid have the same cardinality.
Matroids have been intensively studied due to their rich structure, see \citep{korte2018combinatorial} for a basic introduction in the context of combinatorial optimization, containing numerous examples of matroids.

A \emph{matroid-constrained matching instance} $(G = (A \cup O, E), \succ, \I)$ consists of a matching instance $(G, \succ)$ and an additional matroid $\I$.
We say a matching $M$ in $G$ is \emph{\Iconstr} if $M(A) \in \I$, i.e., the set of matched objects form an independent set of the matroid.
In a matroid-constrained matching instance the set of alternatives thus is the set of {\Iconstr} matchings, with the agents' preferences on the objects inducing preferences on the {\Iconstr} matchings as before.

\smallskip

To illustrate the generality of the matroid-constrained model, we show that it captures a natural variant of the arborescence problem. Consider an arborescence instance $(D=(A\cup\{r\},E),\succ)$, but relax the set of feasible alternatives from $r$-arborescences to all branchings $B$ in $D$. That is, $B$ must remain acyclic, but some agents may have no incoming arc in $B$. This relaxed model was introduce by \citet{kavitha2022popularBranching}, motivated by applications to delegation rules in liquid democracy \cite{SchmidtKraepelin2023PhD}. We can reduce any such branching instance to a matroid-constrained matching instance $(G=(A\cup O,E'),\succ,\mathcal{I})$ as follows. The set of agents stays the same. The object set $O$ corresponds to the directed edges $E$ of $D$. We construct the bipartite edge set $E'$ by connecting each agent $a\in A$ to all objects (directed edges) in $\delta_D^-(a)$. On the object side we impose the graphic matroid: independent sets are those subsets of $O=E$ that form an acyclic set in the underlying undirected graph of $D$. Then, $\mathcal{I}$-constrained matchings in $G$ are in one-to-one correspondence with branchings in $D$. Moreover, bases of the graphic matroid correspond exactly to $r$-arborescences in $D$. If we additionally require, within the matroid-constrained model, that the selected set of objects forms a basis (as in \cite{KMS+25a}), then the framework also captures the arborescence model and, as a special case, the assignment (perfect matching) model studied by \citet{kavitha2022popular}. We discuss resulting open questions in \Cref{sec:discussion}.

\subsection{Solution Concepts} \label{sec:solutionConcepts}

The following solution concepts are defined for any election $(\mathcal{X}, \succ)$, where $\mathcal{X}$ is the set of alternatives and $\succ = (\succ_a)_{a \in A}$ is the profile of agents’ preference relations over $\mathcal{X}$ (given as strict rankings, weak rankings, or partial orders). In our combinatorial settings, $\mathcal{X}$ is either the set of all (matroid-constrained) matchings in a bipartite graph or the set of arborescences in a directed graph. In both cases, the preferences over $\mathcal{X}$ are induced from agents’ preferences over the underlying objects.

\paragraph{Condorcet-Winning Set}

We define a pairwise comparison between a set of alternatives $\mathcal{Y} \subseteq \mathcal{X}$ and an alternative $Z \in \mathcal{X}$.
For this, we let
\begin{align*}
  \varphi(\mathcal{Y},Z) & = \left|\left\{a \in A : \exists \; Y \in \Y \text{ s.th. } Y \succ_a Z \right\}\right|, \text{ and } \\   
  \varphi(Z,\mathcal{Y}) & =  \left|\left\{a \in A : Z \succ_a Y \; \forall \; Y \in \Y \right\}\right|,
\end{align*}
respectively, denote the number of agents preferring the set $\mathcal{Y}$ over $Z$ and vice versa, respectively.
Building upon this, we also define $$\mu(\mathcal{Y},Z) = \varphi(\mathcal{Y},Z) - \varphi(Z, \mathcal{Y}).$$
A set $\mathcal{Y}$ is a \emph{Condorcet-winning} set, if $\mu(\mathcal{Y}, Z) > 0$ for all $Z \in \mathcal{X}$. A set $\mathcal{Y}$ is a \emph{weakly Condorcet-winning} set if $\mu(\mathcal{Y}, Z) \geq 0$ for all $Z \in \mathcal{X}$. We also refer to weakly Condorcet-winning sets as \emph{popular}. The \emph{(weak) Condorcet dimension} of an instance $(\X, \succ)$ is the smallest cardinality of a (weakly) Condorcet-winning set. We remark that when $|\mathcal{Y}| = 1$, our notion of a (weak) Condorcet-winning set coincides with the notion of a (weak) Condorcet winner, as defined in the \Cref{sec:intro}.

\subsubsection*{Pareto Optimality}

\newcommand{\Z}{\ensuremath{\mathcal{Z}}}

Under weak rankings, the standard notion of Pareto optimality for a single alternative $X \in \mathcal{X}$ is as follows: $X$ is Pareto optimal if no alternative $Y$ Pareto-dominates it, i.e., if there is no $Y$ such that every agent weakly prefers $Y$ to $X$ and at least one agent strictly prefers $Y$ to $X$. We generalize this notion in two directions: we extend it from single alternatives to sets of alternatives, and we allow preferences given by partial orders.

A set of alternatives $\Y \subseteq \X$ \emph{Pareto-dominates} another set of alternatives $\Z \subseteq \X$ if 
\begin{itemize}
    \item $|\Y| \leq |\Z|$,
    \item for every agent $a \in A$ there is $Y \in \Y$ with $Y \not\prec_a Z$ for all $Z \in \N$,
    \item there is an agent $a \in A$ and a $Y \in \Y$ such that $Y \succ_a Z$ for all $Z \in \N$.
\end{itemize}
A set of alternatives is \emph{Pareto-optimal} if it is not Pareto-dominated by any other set of alternatives. For brevity, we often write \emph{dominates} instead of \emph{Pareto-dominates}. 

First, consider the above definition in the case of weak rankings. Then comparing two sets aligns with how an agents evaluate sets in the definition of a Condorcet winning set: a set $\mathcal{Y}$ is compared to a set $\mathcal{Z}$ by comparing their most-preferred element in $\mathcal{Y}$ to their most-preferred element in $\mathcal{Z}$.

Now consider singleton sets and partial-order preferences. In this case, the definition above treats the condition $Y \not\prec_a Z$ as a weak preference for $Y$ over $Z$ since either $Y \succ_a Z$, or $Y \sim_a Z$, where in the latter case $a$ is indifferent between $Y$ and $Z$.%
\footnote{We remark that an alternative definition arises when interpreting $\sim$ as incomparability rather than indifference. To formalize this, one may work with an explicit weak relation $\succeq_a$ and distinguish true indifference ($Y \succeq_a Z$ and $Z \succeq_a Y$) from incomparability (neither $Y \succeq_a Z$ nor $Z \succeq_a Y$), declaring that $Y$ dominates $Z$ only if all agents report $Y \succeq_a Z$. This leads to a stricter notion of Pareto-dominance and hence a looser notion of Pareto-optimality compared to the definition given in the text. We adopt the latter to align with our voting interpretation under partial orders, as is common in the popularity literature (e.g., \cite{kavitha2022popularBranching}).}

This extension of Pareto optimality to partial orders has the---perhaps surprising---property of not being guaranteed to exist. For a simple yet insightful example in the case of matchings, consider the case where $G = (A \cup O, E)$ is the complete graph with $3$~agents $A = \{a_0, a_1, a_2\}$ and $3$~objects $O = \{o_0, o_1, o_2\}$.
When all agents have the same preferences, namely $o_0 \succ o_2$ but being indifferent among any other pair of objects, there does not exist a Pareto-optimal matching: First, any such matching $M$ would need to match all agents. By symmetry, we can assume without loss of generality that $M(a_i) = o_i$. But then the matching $M'$ with $M'(a_i) = o_{i+1}$ (with $o_3 = o_0$) Pareto-dominates $M$. We remark that agents' preferences in this example are even semiorders \cite{luce1956semiorders}.

\subsubsection*{Top-Choice Alternative}

We call an alternative $X \in \X$ a \emph{top-choice alternative} if, for every agent $a \in A$, there is no alternative $Y \in \X$ with $Y \succ_a X$. In other words, $X$ is undominated for each agent. 

Top-choice alternatives rarely exist, e.g., in the context of matchings under strict preferences, a top-choice matching is one in which every agent receives their favorite object. If we add matroid constraints the notion becomes slightly more subtle: in a top-choice $\mathcal{I}$-constrained matching each agent receives their favorite object among those they can obtain in some $\mathcal{I}$-constrained matching.

\section{Warm-Up: Pareto-optimality and the Condorcet Dimension for Matchings under Strict Preferences}\label{sec:warmup}

We start by considering the special case where agents have strict preferences and there is no matroid constraint.
We show that Pareto-optimality implies strong popularity in this case (except for cases where there exists a top-choice matching, which then is strongly popular by itself).

\begin{theorem}\label{thm:warm-up}
    Let $(G, \succ)$ be a matching instance with strict preferences. Let $\M = \{M_1, M_2\}$ be a Pareto-optimal set of two matchings in $G$. 
    Then $\M$ is popular. Moreover, one of the following holds: 
    \begin{enumerate}
        \item $\M$ is strongly popular, 
        \item or there exists a top-choice matching. 
    \end{enumerate}
\end{theorem}

{An important ingredient in the proof of \cref{thm:warm-up} is the following lemma, which allows us to compare the number of agents favoring $\M$ over an alternative matching $N$ (later called \emph{blue agents}) to those that favor $N$ over $\M$ (later called \emph{red agents}) and those that are indifferent (later called \emph{grey agents}), by means of an appropriately constructed branching. 
The proof of the more general \cref{thm:Pareto-implies-Condorcet} in the next section will run along similar lines and also make use of \cref{lem:colored-branching}.}

\begin{lemma}\label{lem:colored-branching}
    Let $A$ be a finite set and $A = A_{\text{R}} \cup A_{\text{B}} \cup A_{\text{G}}$ be a partition of $A$ into red, blue, and grey elements.
    Let $B$ be a branching on $A$ with the following properties:
    \begin{enumerate}
        \item Every 
        leaf in a non-singleton component is blue.
        \item Every red node has outdegree at least $2$.
    \end{enumerate}
    Then $|A_{\text{R}}| \leq |A_{\text{B}}|$, with equality only if $A_{\text{R}} = \emptyset = A_{\text{B}}$.
\end{lemma}

\begin{proof}
    Consider any non-singleton component (i.e., out-tree) $T$ of $B$, and let $r$, $b$, and $g$, respectively, denote the number of red, blue, and grey nodes in $T$.
    We show that $T$ contains more blue than red nodes.
    Note that every node in $T$ except the root has indegree $1$, and thus the total sum of in-degrees of the nodes in $T$ is $r + b + g - 1$.
    Because every red node has out-degree at least $2$ and every grey node in $T$ has out-degree at least $1$ (because all leafs of $T$ are blue), the total sum of out-degrees of the nodes in $T$ is at least $2r + g$.
    As the total sum of in-degrees must equal the total sum of out-degrees, we obtain $r + b + g - 1 \geq 2r + g$, and hence $b \geq r + 1 > r$.
    Thus, every non-singleton component of $B$ contains strictly more blue then red nodes.
    
    Because all red nodes are contained in non-singleton components (due to their out-degree being $2$), we conclude that there are at least as many blue nodes as there are red nodes, i.e., $|A_{\text{R}}| \leq |A_{\text{B}}|$, with equality only if all nodes are grey singletons,  i.e., $A_{\text{R}} = \emptyset = A_{\text{B}}$.
\end{proof}

We now show how to prove \cref{thm:warm-up} using \cref{lem:colored-branching}.

\begin{proof}[Proof of \cref{thm:warm-up}]
First note that we can assume without loss of generality that every agent $a \in A$ is assigned an object in at most one of the two matchings $M_1$ and $M_2$, because under strict preferences, Pareto-optimality only depends on the best object that $a$ is assigned to among $M_1(a)$ and $M_2(a)$, so we can drop the other assignment.

Now let $N$ be any matching in $G$. 
We color agents as follows: $a$ is blue if $\M \succ_a N$, red if $N \succ_a \M$, and grey if $a$ is indifferent between $\M$ and $N$. 
Note that $\varphi(\M, N)$ equals the number of blue agents and that $\varphi(N, \M)$ equals the number of red agents. 
Below, we construct a branching on $A$ meeting the criteria of \cref{lem:colored-branching}.
This implies that the blue agents strictly outnumber the red agents (and thus $\mu(\M, N) > 0$), unless all agents are grey. 
Note that the latter case is only possible if~$N(a) = \M(a)$ for all $a \in A$, i.e., $N$ assigns $a$ to the same object that it receives in $\M$.
But then~$N$ must be a top-choice matching, as otherwise, there exists a matching $N'$ and an agent $a \in A$ such that $N' \succ_a N$, which implies that $\{N, N'\}$ Pareto-dominates $\M$, a contradiction.

To construct the branching, we first show that for every red agent $a$, the object assigned to $a$ in~$N$ is assigned to two different agents (colored red or blue) in $M_1$ and $M_2$, respectively.
\begin{claim}\label{clm:warm-up-successors}
    For any red agent $a$, there are two distinct agents $b^1_a$ and $b^2_a$, colored red or blue, such that $M_1(b^1_a) = N(a) = M_2(b^2_a)$.
\end{claim}
\begin{proof}
Let $a$ be a red agent and let $o := N(a)$ be the object assigned to $a$ in $N$.
If $o$ is not assigned in $M_1$, then we can improve $\M$ by assigning $o$ to $a$ in $M_1$, contradicting Pareto-optimality.
Thus there must be agent $b^1_a \in A$ with $M_1(b^1_a) = N(a)$.
Note that $N(b^1_a) \neq N(a) = M_1(b^1_a)$, and therefore $b^1_a$ is either blue or red.
By the same reasoning there must be a blue or red agent $b^2_a \in A$ with $M_2(b^2_a) = N(a)$.
Moreover, $b^1_a \neq b^2_a$, as every agent is assigned only in one of the two matchings $M_1, M_2$ by our earlier assumption.
\end{proof}

We construct the branching $B$ by including the arcs from each red agent $a$ to the two agents $b^1_a$ and $b^2_a$ given by \cref{clm:warm-up-successors}, i.e.,
$B := \{(a, b^i_a) : a \text{ is  red}, i \in \{1, 2\}\}$, see \cref{fig:warmup} for an example.
Note that the indegree of every agent $a$ in $B$ is at most one, as no agent is matched in both $M_1$ and $M_2$ and there is at most one agent that receives the same object in $N$ that $a$ receives in $M_1$ or $M_2$.
Moreover, every non-singleton component of $B$ contains only red and blue nodes, of which all red nodes have outdegree~$2$ and all blue nodes have out-degree $0$.
It thus remains to show that $B$ contains no cycles.

Assume by contradiction that $B$ contains a cycle $C$.
Let $A(C)$ denote the set of nodes on the cycle.
Note that all nodes in $A(C)$ are red (as blue and grey agents do not have out-arcs) and that all arcs in $C$ are of the type $(a, b^i_a)$ for some $a \in A(C)$ and $i \in \{1, 2\}$.
Let $A_1(C) := \{a \in A(C) : (a, b^1_a) \in C\}$ and $A_2(C) := \{a \in A(C) : (a, b^2_a) \in C\}$.
We construct two matchings $M_1'$ and $M_2'$ as follows:
\begin{align*}
    M_1'(a) := 
    \begin{cases}
        N(a) & \text{ for } a \in A_1(C)\\
        \emptyset & \text{ for } a \in A_2(C)\\
        M_1(a) & \text{ for } a \in A \setminus A(C)
    \end{cases}
    \quad
    M_2'(a) := 
    \begin{cases}
        N(a) & \text{ for } a \in A_2(C)\\
        \emptyset & \text{ for } a \in A_1(C)\\
        M_2(a) & \text{ for } a \in A \setminus A(C)
    \end{cases}
\end{align*}
Note that $M'_1$ and $M'_2$ are indeed matchings:
Assume by contradiction that $M'_i(a) = M'_i(a')$ for $i \in \{1, 2\}$ and $a, a' \in A$.
This is only possible if $a \in A_i(C)$, $a' \in A \setminus A(C)$ and 
$M'_i(a) = N(a) = M_i(a')$ (up to swapping $a$ and $a'$).
But this implies $a' = b^i_a$, which together with $a \in A_i(C)$ yields $(a, a') \in C$, i.e., $a' \in A(C)$, a contradiction.
Thus, $M_1'$ and $M_2'$ are matchings and $\M' := \{M_1', M_2'\}$ Pareto-dominates $\M$ because in $\M'$ the red agents $a \in A(C)$ all receive their more preferred object $N(a)$, while all agents in $A \setminus A(C)$ receive the same item as in $\M$.
\end{proof}

{To apply \Cref{thm:warm-up} and obtain an upper bound of $2$ for the weak Condorcet dimension under strict rankings, it suffices to show that a Pareto-optimal set of size $2$ always exists. For example, one can construct such a set via a simple round-robin procedure: duplicate each object and let agents, in turn, pick their favorite  remaining adjacent object. This yields a Pareto-optimal $2$-matching $M$, which decomposes into a Pareto-optimal set of matchings $\M$. To see why $M$ is Pareto-optimal, consider any other $2$-matching $N$. Then, some agent must receive an object $N(a) \neq M(a)$ which was picked after their turn. Thus, $M(a) \succ_a N(a)$ and $N$ cannot Pareto-dominate $M$. In \Cref{sec:pareto-implies-condorcet-matroids} we discuss more ways to find Pareto-optimal sets.}

\begin{figure}
    \centering
    \includegraphics[width=0.7\linewidth]{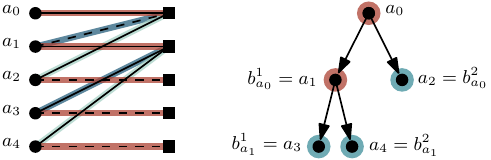}
    \caption{The construction of the branching used in the proof of \cref{thm:warm-up}. The graph on the left depicts the matching instance, with dark-blue edges belonging to $M_1$, light-blue edges belonging to $M_2$, and red edges belonging to $N$.
    The resulting branching is depicted on the right.}
    \label{fig:warmup}
\end{figure}

\section{Pareto-optimality and the Condorcet Dimension for Matroid-Constrained Matchings}
\label{sec:pareto-implies-condorcet-matroids}

In this section, we extend \cref{thm:warm-up} in two directions: First, we allow a matroid-constraint as described in \cref{sec:prelim-matroids}. Second, we allow general partial-order preferences for the agents.
We show that Pareto-optimality still implies popularity in this general setting.

\begin{theorem}\label{thm:Pareto-implies-Condorcet}
    Let $(G, \mathcal{I}, \succ)$ be a matroid-constrained matching instance with partial-order preferences.
    Then any Pareto-optimal set $\mathcal{M}$ of $\mathcal{I}\!$-constrained matchings in $G$ with~$|\mathcal{M}| \geq 2$ is popular.\\
    Moreover, if preferences are strict rankings, one of the following holds:
    \begin{itemize}
        \item Every Pareto-optimal set $\mathcal{M}$ of $\mathcal{I}\!$-constrained matchings in $G$ with~$|\mathcal{M}| \geq 2$ is strongly popular,
        \item or there exists a top-choice $\mathcal{I}\!$-constrained matching in $G$.
    \end{itemize}
\end{theorem}

\cref{thm:Pareto-implies-Condorcet} is a consequence of \cref{lem:Pareto-vote-count} (in combination with \cref{lem:colored-branching}) stated and proven below. 
Before we discuss this proof, let us first observe some consequences of the theorem.
To use \cref{thm:Pareto-implies-Condorcet} to bound the (weak) Condorcet dimension for matroid-constrained matchings, we use the fact that when preferences are weak rankings, Pareto-optimal sets of any cardinality exist and can be efficiently computed by solving a maximum-weight matroid intersection problem.

\begin{restatable}{lemma}{restateLemComputingPareto}\label{lem:computing-pareto}
    There is an algorithm, that given a matroid-constrained matching instance $(G, \mathcal{I}, \succ)$ with weak-ranking preferences, where $\mathcal{I}$ is given by an independence oracle, computes in polynomial time a Pareto-optimal set $\mathcal{M}$ of $\mathcal{I}\!$-constrained matchings in $G$ with~$|\mathcal{M}| \leq 2$.
\end{restatable}

{The proof of \cref{lem:computing-pareto} is given in \cref{app:matroids-computing}. For unconstrained matchings, finding a Pareto-optimal set is even simpler as a single maximum-weight bipartite matching computation suffices: after duplicating objects, one can set weights to agents’ ranks. In both cases, our proof crucially uses the fact that the preferences are weak rankings.
For general partial-order preferences, the existence of Pareto-optimal sets of (unconstrained) matchings of constant size cannot be guaranteed, as we show in \cref{sec:partial-orders}.} We obtain the following corollaries from \cref{thm:Pareto-implies-Condorcet,lem:computing-pareto}.

\begin{corollary}\label{cor:Condorcet-weak-order}
    If preferences are weak rankings, a popular set of at most two $\mathcal{I}\!$-constrained matchings exists and can be computed in polynomial time when $\I$ is given by an independence oracle.
\end{corollary}

\begin{corollary}\label{cor:Condorcet-linear-order}
    If preferences are linear orders, a strongly popular set of at most two $\mathcal{I}\!$-constrained matchings exists and can be computed in polynomial time when $\I$ is given by an independence oracle.
\end{corollary}

\cref{cor:Condorcet-linear-order} follows directly from \cref{thm:Pareto-implies-Condorcet,lem:computing-pareto}.
To obtain \cref{cor:Condorcet-linear-order}, we additionally use the fact that the existence of a top-choice matching can be easily checked (either assigning each agent their top choice object results in an~$\mathcal{I}\!$-constrained matching, or no top-choice matching exists) and if it exists yields a strongly popular set of size one.

\subsection{Proof of \cref{thm:Pareto-implies-Condorcet}}
To prove \cref{thm:Pareto-implies-Condorcet}, we establish the following lemma and explain how it implies the theorem. The remainder of this section then will be devoted to the proof of the lemma (see overview below).

\begin{lemma}\label{lem:Pareto-vote-count}
    Let $(G, \I, \succ)$ be a matroid-constrained matching instance.
    For a Pareto-optimal set of {\Iconstr} matchings  $\M$ in $G$ and an {\Iconstr} matching $N$ in $G$, let
    \begin{align*}
        A_+ & := \{a \in A : N \succ_a M \text{ for all } M \in \mathcal{M}\} \text{ and }\\
        A_- & := \{a \in A : M \succ_a N \text{ for some } M \in \mathcal{M}\}.
    \end{align*}
    Then there exists a branching $B$ on the set of agents $A$ such that
    \begin{enumerate}
        \item every 
        leaf in a non-singleton component of $B$ is in $A_-$,
        \item every $a \in A_+$ has out-degree $|\M|$.
    \end{enumerate}
\end{lemma}

\begin{proof}[Proof of \cref{thm:Pareto-implies-Condorcet} using \cref{lem:Pareto-vote-count}]
    Let $\M$ be a Pareto-optimal set of {\Iconstr} matchings in $G$ and let $N$ be an {\Iconstr} matching in $G$. Note that the branching $B$ given by \cref{lem:Pareto-vote-count} fulfills the requirements of \cref{lem:colored-branching} when coloring the nodes in $A_+$ red, the nodes in $A_-$ blue, and all other nodes grey.
    Applying \cref{lem:colored-branching} yields $|A_+| \geq |A_-|$, with equality only when $A_+ = \emptyset = A_-$.
    This directly implies the first statement of \cref{thm:Pareto-implies-Condorcet}, namely that any Pareto-optimal set $\mathcal{M}$ of $\mathcal{I}\!$-constrained matchings in $G$ with~$|\mathcal{M}| \geq 2$ is popular, as $\varphi(N, \M) = |A_+|$ and $\varphi(\M, N) = |A_-|$. 
    
    To see that \cref{lem:Pareto-vote-count} also implies the stronger statement for linear-order preferences, assume by contradiction that $\mathcal{M}$ is not strongly popular.
    Then there must be~$N$ such that $\varphi(\M, N) = \varphi(N, \M)$, which implies $|A_+| = |A_-|$ and therefore $A_+ = \emptyset = A_-$ as described above.
    But under linear-order preferences this is equivalent to $N(a) = \max_{\succ_a} \{M(a) : M \in \M\}$ for all $a \in A$.
    Now, if $N$ is not a top-choice {\Iconstr} matching, then there is an {\Iconstr} matching $M'$ and an agent $a \in A$ such that $M'(a) \succ_a N(a)$.
    But then the set $\{N, M'\}$ Pareto-dominates $\M$, contradicting the Pareto-optimality of the latter.
\end{proof}

\subsubsection*{Overview of \cref{sec:pareto-implies-condorcet-matroids}}
In the remainder of this section, we prove \cref{lem:Pareto-vote-count}. 
Recall that in the proof of \cref{thm:warm-up} an arc $(a, b)$ in the branching indicated that $N$ assigns some item to agent $a$ that is assigned to agent $b$ in one of the matchings in $\M$.
In the presence of matroid-constraints, this straightforward construction is no longer sufficient: A ``red'' agent 
$a \in A_+$ may be assigned an item $o$ in $N$ that is not assigned to any agent in any $M \in \M$ without yielding an immediate contradiction to Pareto-optimality, as the matroid constraint may prevent us from assigning $o$ to $a$ in any matching $M \in \M$.

To overcome this issue, we construct an exchange graph (\cref{sec:matroid-exchange-graph}) on the agents based on the fundamental circuits generated by the 
objects assigned to them.
We then apply the bijective basis exchange property of matroids (\cref{sec:matroid-bijective-exchange}) to carefully constructed paths in the exchange graph leading from any ``red'' agent to a ``blue'' agent.
In \cref{sec:matroid-branching}, we show crucially that the constructed paths cannot intersect (\cref{lem:B-indegree}) and cannot form cycles (\cref{lem:B-nocycles}), from which we obtain the desired branching.

\subsection{The Objects $\emptyset_a$ and the Exchange Graph}
\label{sec:matroid-exchange-graph}

Recall that a basis of the matroid is an inclusionwise maximal independent set.
We modify the instance so that we can assume without loss of generality that the set of matched objects forms a basis.
We then introduce an exchange graph on the agents based on the concept of fundamental circuits of a basis.

\subsubsection*{Modified instance with objects $\emptyset_a$}
For notational convenience we apply the following transformation to the instance, after which all maximal matroid-constrained matchings $M$ have the property that $M(A)$ is a basis of~$\I$.
For this, we first introduce an object $\emptyset_{a}$ for every agent $a \in A$, and add the edge~$(a, \emptyset_a)$ to $E$.
We moreover introduce the preferences $o \succ_a \emptyset_a$ for all $o \in O$.
We replace the original matroid $\I$ by $$\I_{\emptyset} := \{I \cup \{\emptyset_a : a \in S\} : I \in \I, S \subseteq A, |I| + |S| \leq |A|\}.$$
This is again a matroid (it is the direct sum of the original matroid $\I$ and the free matroid on $\{\emptyset_a : a \in A\}$, truncated at $|A|$; all these operations are well-known to preserve the matroid properties).
Note that every {\Iconstr} matching in the original instance corresponds to an $\I_{\emptyset}$-constrained matching in the modified instance and vice versa, by assigning each unmatched agent $a$ to $\emptyset_a$.
Moreover, in the modified instance, any inclusionwise maximal $\I_{\emptyset}$-constrained
matching $M$ assigns all agents to an object (if there was an unmatched agent $a$, we could extend $M$ by assigning $a$ to $\emptyset_a$).
Hence $M(A)$ has cardinality $|A|$ and is thus a basis of $\I_{\emptyset}$.

For the remainder of \cref{sec:pareto-implies-condorcet-matroids}, we shall assume that the above transformations have been applied, i.e., there is $\emptyset_a \in O$ for $a \in A$ as described above, incident exclusively to agent $a$ in $G$, and $\I$ has been replaced by $\I_{\emptyset}$ as described above. 
This allows us to assume henceforth, without loss of generality, that $M(A)$ is a basis of $\I$ for all $M \in \M \cup \{N\}$.

\subsubsection*{Fundamental circuits and the exchange graph}
For a basis $X$ of $\I$ and $o \in O \setminus X$, the \emph{fundamental circuit of $x$ for $X$} is $C(X, o) := \{x \in X : X \setminus \{x\} \cup \{o\} \in \I\}$, i.e., the set of all elements of $X$ that can be exchanged for $o$ to obtain a new basis.
Given two {\Iconstr} matchings $M, N$, the \emph{exchange graph from $N$ into $M$} is the directed graph $D_{N \rightarrow M} = (A, E_{N \rightarrow M})$ on the set of agents $A$ with arc set $$E_{N \rightarrow M} := \{(a, b) : M(b) \in C(M(A), N(a))\},$$ i.e., there is an arc from agent $a$ to agent $b$ if the object assigned to $b$ in $M$ is on the unique circuit closed by adding the object $N(a)$ to the basis $M(A)$.

Note that an arc $(a, b)$ in the exchange graph from $N$ to $M$ indicates that we can change $M$ by assigning agent $a$ the object $N(a)$, as long as we unassign object $M(b)$ from agent $b$.
This idea effectively generalizes to cycles and paths in the exchange graph, which we formalize in \cref{lem:exchange-cycle,cor:exchange-path} below.
Both these results are a consequence of a result by \citet{frank1981weighted}; see \cref{app:matroids-exchange} for a complete proof. 

\begin{restatable}{lemma}{restateLemExchangeCycle}\label{lem:exchange-cycle}
    Let $M, N$ be two {\Iconstr} matchings in $G$ and let $C$ be a cycle in $D_{N \rightarrow M}$.
    Let $A(C)$ denote the set of agents on $C$ and let $a' \in A(C)$.
    Then there is $A' \subseteq A(C)$ with $a' \in A'$ such that $M'$ defined by $M'(a) := N(a)$ for all $a \in A'$ and $M'(a) := M(a)$ for all $a \in A \setminus A'$ is an {\Iconstr} matching in $G$.
\end{restatable}

\begin{restatable}{corollary}{restateCorExchangePath}\label{cor:exchange-path}
    Let $M, N$ be two {\Iconstr} matchings in $G$ and let $P$ be an $a'$-$b'$-path in $D_{N \rightarrow M}$.
    Let $A(P)$ denote the set of agents on $P$.
    Then there is $A' \subseteq A(P)$ with $a', b' \in A'$ such that $M'$ defined by $M'(a) := N(a)$ for all $a \in A' \setminus \{b'\}$, $M'(b') = \emptyset_{b'}$, and $M'(a) := M(a)$ for all $a \in A \setminus A'$ is an {\Iconstr} matching in $G$.
\end{restatable}

\subsection{Applying Bijective Basis Exchange}
\label{sec:matroid-bijective-exchange}

    Note that in the exchange graph defined above each agent may have a large in- and out-degree, corresponding to a large number of choices for determining which object assigned to some agent $b$ in $M$ had to be given up in $N$ in order to assign a potentially better object to $a$.
    We now apply the \emph{bijective basis exchange property} of matroids, stated in \cref{lem:bijective-basis exchange} below, to show that there is a consistent set of choices for each $M \in \M$ obtaining a subgraph $D_M$ of the exchange graph $D_{N \rightarrow M}$ that consists of node-disjoint cycles. 
    We will then use subpaths on these cycles to construct the desired branching.

    \begin{lemma}[Bijectve basis exchange property \citep{brualdi1969comments}]\label{lem:bijective-basis exchange}
        Let $(O, \mathcal{I})$ be a matroid and let $I_1, I_2$ be two bases of $\mathcal{I}$.
        Then there is a bijection $f: I_1 \rightarrow I_2$ such that for every $o \in I_1$ the set $I_1 \setminus \{o\} \cup \{f(o)\}$ is a basis of $\mathcal{I}$.
    \end{lemma}
    For $M \in \M$, let $f_{M} : N(A) \rightarrow M(A)$ denote the bijection given by \cref{lem:bijective-basis exchange} for the bases $M(A)$ and $N(A)$ with $f_{M}(o) \in C(M(A)\cup \{o\})$ for all $o \in N(A)$.
    Consider the digraph $D_{M} = (A, E_M)$ with $$E_M := \{(a, b) \in A \times A : M(b) = f_{M}(N(a))\}.$$
    Note that $D_M$ is the subgraph of $D_{N \rightarrow M}$ resulting from the restriction to arcs corresponding to pairs of agents whose objects in $M$ and $N$ are paired in the bijection~$f_M$.

    Because $f_M$ is a bijection, $E_M$ contains exactly one outgoing arc and exactly one incoming arc for each agent (note that some agents may be singletons with a loop).  
    Therefore, $E_{M}$ consists of a union of pairwise node-disjoint directed cycles (some of which may be loops), and each node appears on exactly one such cycle.
    The following lemma, which is an immediate consequence of \cref{lem:exchange-cycle}, shows that any such cycle containing an agent from $A_+$ must also contain an node from $A_-$.

    \begin{lemma}\label{lem:AplusCycleAminus}
        Let $a_+ \in A_+$ and let $M \in \mathcal{M}$. Then the unique cycle in $D_{M}$ containing~$a_+$ also contains some node from $A_-$ (in particular, $a_+$ cannot be base of a loop in $D_{M}$).
    \end{lemma}
    \begin{proof}
        Let $C$ be the unique cycle that contains $a_+$ in $D_{\phi}$ and let $A(C)$ be the nodes visited by $C$.
        By contradiction assume that $A(C) \cap A_- = \emptyset$.
        We apply \cref{lem:exchange-cycle} to the cycle $C \subseteq E_{N \rightarrow M}$ and the agent $a_+ \in C$ to obtain a subset of agents $A' \subseteq A(C)$ with~$a_+ \in A'$ and an~{\Iconstr} matching $M'$ with 
        $$M'(a) = \begin{cases}
            N(a) & \text{ for all } a \in A',\\
            M(a) & \text{ for all } a \in A \setminus A'.
        \end{cases}$$
        Note for all $\bar{M} \in \mathcal{M}$ it holds that $M'(a_+) = N(a_+) \succ_{a_+} \bar{M}(a)$, because $a_+ \in A' \cap A_+$, and that $M'(a) = N(a) \not\prec_a \bar{M}(a)$ for all $a \in A'$, because $A' \cap A_- = \emptyset$.
        Therefore, $\mathcal{M}' := \mathcal{M} \setminus \{M\} \cup \{M'\}$ dominates $\mathcal{M}$, contradicting Pareto-optimality of $\mathcal{M}$.
    \end{proof}

\subsection{Constructing the Branching}
\label{sec:matroid-branching}

    We will now construct a branching.
    For $M \in \mathcal{M}$, and $a \in A_+$, let $C^M_a$ be the unique cycle containing $a$ in the digraph $D_M$, as constructed above.
    By \cref{lem:AplusCycleAminus}, $C^M_a$ must contain a node from $A_-$, and we let $b^M_a \in A_-$ denote the first agent in $A_-$ that one encounters when traversing $C^M_a$ starting from $a$.
    Let $P^M_a$ denote the $a$-$b^M_a$-path in~$C^M_a$.
    
    Note that for $a, a' \in A_+$, the paths $P^M_a, P^M_{a'}$ are either node-disjoint, or one is contained in the other (which happens when $a$ and $a'$ are on the same cycle in $D_M$ and not separated by nodes from~$A_-$).
    Let $A^{M}_+ \subseteq A_+$ denote the set of agents $a \in A_+$ such that~$P^M_a$ is not contained in $P^M_{a'}$ for any $a' \in A_+ \setminus \{a\}$.
    Let $B := \bigcup_{M \in \M, a \in A_+^M} P^M_a$. 
    See \cref{fig:matroid-branching} for an example of the construction.

    We will show that $B$ is a branching.
    We first establish the following lemma, showing the paths $P^M_a$ for different $M \in \M$ and $a \in A^M_+$ are node-disjoint except for their starting nodes.
    In particular, every node in $A$ therefore is entered by at most one such path and hence the indegree of any node in $B$ is no larger than $1$.

    \begin{figure}
        \centering
        \includegraphics[width=\linewidth]{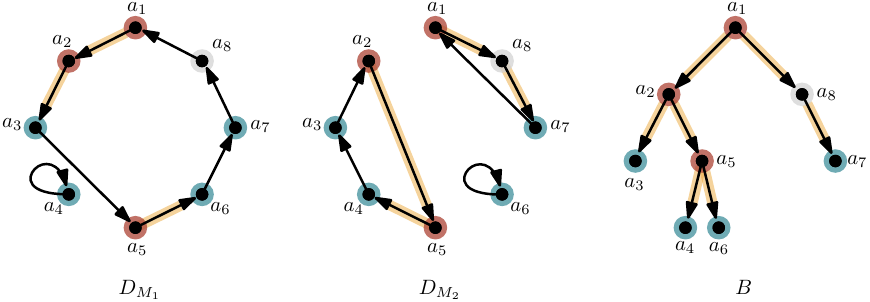}
        \caption{%
        Example for the construction of the branching $B$ for matroid-constrained matchings $\M = \{M_1, M_2\}$ and agents $A = \{a_1, \dots, a_8\}$.
        The agents in $A_+ = \{a_1, a_2, a_5\}$ are colored red, the agents in $A_- = \{a_3, a_4, a_6, a_7\}$ are colored blue, and the remaining agent $a_8$ is color grey.
        The left and middle figure depict the digraphs $D_{M_1}$ and $D_{M_2}$ respectively, with the arcs on paths $P^{M}_a$ highlighted. 
        Note that $b^{M_1}_{a_1} = a_3$ is the first blue agent encountered when traversing the cycle in $D_{M_1}$ starting from $a_1$, yielding the path $P^{M_1}_{a_1} = \{(a_1, a_2), (a_2, a_3)\}$.
        Note that $P^{M_1}_{a_2} = \{(a_2, a_3)\} \subseteq P^{M_1}_{a_1}$ and $P^{M_1}_{a_5} = \{(a_5, a_6)\}$. Thus $A^{M_1}_+ = \{a_1, a_5\}$.
        Similarly, $b^{M_1}_{a_2} = a_4$ with $P^{M_2}_{a_2} = \{(a_2, a_5), (a_5, a_4)\}$, which contains $P^{M_2}_{a_5} = \{(a_5, a_4)\}$.
        Hence $A^{M_2}_+ = \{a_1, a_2\}$.
        The figure to the right indicates the resulting matching $B = P^{M_1}_{a_1} \cup P^{M_1}_{a_5} \cup P^{M_2}_{a_1} \cup P^{M_2}_{a_2}$.
        }
        \label{fig:matroid-branching}
    \end{figure}

    \begin{restatable}{lemma}{restateLemBindegree}\label{lem:B-indegree}
        For every $a' \in A$, there is at most one $M \in \M$ and $a \in A^{M}_+ \setminus \{a'\}$ such that $a'$ is contained in $P^M_a$.
        In particular, $a'$ has in-degree at most $1$ in $B$.
    \end{restatable}
    \begin{proof}[Sketch of proof]
        By contradiction assume that for $a' \in A$ there are $M_1, M_2 \in \mathcal{M}$ with $M_1 \neq M_2$ and $a_1, a_2$ such that $a'$ is contained in both in~$P^{M_1}_{a_1}$ and $P^{M_2}_{a_2}$.
        Then we can construct a set of {\Iconstr} matchings $\M'$ Pareto-dominating $\M$, yielding a contradiction.

        Without loss of generality, we may assume $M_2(a') \not\prec_{a'} M_1(a')$ (otherwise swap $M_1$ and $M_2$).
        Hence we can apply \cref{cor:exchange-path} to the $a_1$-$a'$-subpath $P_1$ of $P^{M_1}_{a_1}$ to obtain an improved matching $M'_1$ in which some agents from $A_+$ receive the more favored object $N(a)$ instead of $M_1(a)$, while the only agent who is worse off in $M'_1$ ais $a'$ with $M'_1(a') = \emptyset$.
        Despite the latter, replacing $M_1$ by $M'$ in $\M$ results in a set of matchings $\M$ Pareto-dominating $\M$, as $M_2(a') \not\prec_{a'} M_1(a')$ prevents $a'$ from strictly preferring $\M$ over $\M'$.
        See \cref{app:matroids-branching-indegree} for a complete and formal proof.
    \end{proof}

    We further show that $B$ does not contain any cycles, implying that $B$ is indeed a branching.
    
    \begin{restatable}{lemma}{restateLemBnoCycles}\label{lem:B-nocycles}
        There is no directed cycle contained in $B$.
    \end{restatable}
    \begin{proof}[Sketch of proof]
        Given a cycle $C$ in $B$, we once more construct a set of {\Iconstr} matchings $\M'$ that Pareto-dominates $\M$, yielding a contradiction.
        Similar to the proof of \cref{thm:warm-up}, we achieve this by changing the assignment of some agents whose out-arc on $C$ comes from a path $P^M_a$ for some $M \in \M$ by replacing $M(a)$ with $N(a)$.
        However, we have to be careful in order not to violate the matroid constraint in any matching.
        For this, we employ a generalization of \cref{cor:exchange-path} that shows that reassignment of agents on multiple paths is possible under certain conditions, which can be met by iteratively shortening the cycle.
        The complete details of the proof, including a proof of the aforementioned generalization of \cref{cor:exchange-path}, can be found in \cref{app:matroids-branching-no-cycle}.
    \end{proof}
    
    It remains to show that all 
    leafs in non-singleton components in $B$ are in $A_-$ and that all nodes in $A_+$ have out-degree $|\M|$.
    For this, recall that $B$ is the union of paths $P^M_a$, each of which ends with a node from $A_-$. 
    Thus, all 
    leafs in non-singleton components in $B$ are in $A_-$.
    Moreover, every node in $a \in A_+$ is contained in $|\M|$ paths $P^M_a$.
    Note that $a$ is not the final node in any of these paths 
    since there is an out-arc $(a, b_M) \in P^M_a$ for each $M \in \M$.
    By \cref{lem:B-indegree}, we have $b_M \neq b_{M'}$ for $M \neq M'$ and hence each path contributes a different out-arc to $a$, i.e., every $a \in A_+$ has out-degree $|\M|$ in $B$.
    This completes the proof of \cref{lem:Pareto-vote-count} and therefore also of \cref{thm:Pareto-implies-Condorcet}.

\section{Bounds on the Condorcet Dimension under Partial-Order Preferences}
\label{sec:partial-orders}

While \cref{thm:Pareto-implies-Condorcet} establishes that Pareto-optimality implies popularity for sets of (matroid-constrained) matchings even under partial-order preferences, this does not immediately yield a bound on the Condorcet dimension for these general preference structures.
The reason is that---in contrast to the case of weak rankings---the existence of Pareto-optimal solutions is no longer guaranteed under partial-order preferences (see \Cref{sec:prelim} for an example).

Indeed, in \cref{sec:partial-order-matchings-lb} we provide a family of matching instances in which the Condorcet dimension (and thus also the size of a minimum-cardinality Pareto-optimal set of matchings) is at least $\sqrt{|A|}$.
This bound is in fact tight (up to a factor of $\sqrt{2}$), which we show in \cref{sec:partial-order-matchings-ub} by providing an efficient algorithm for constructing a popular set of matchings of size at most $\lceil \sqrt{2|A|} \rceil$ under partial-order preferences.
In \cref{sec:partial-order-matchings-matroids}, we show that when additionally allowing matroid constraints, the Condorcet dimension can be as high as $|A|$, which trivially is tight because $|A|$ matchings suffice to match every agent to an object that is undominated w.r.t.~their preferences.

\subsection{A Lower Bound on the Weak Condorcet Dimension for Matchings on Partial-order Preferences}
\label{sec:partial-order-matchings-lb}

We construct a sequence of instances of increasing size with Condorcet dimension $k + 1$ and $k^2 + k + 1$ agents, establishing that the Condorcet dimension for matchings is lower bounded by $\sqrt{|A|}$.

\begin{theorem}
    For every $k \in \mathbb{N}$, there exists a matching instance with $k^2 + k + 1$ agents with partial order preferences that has weak Condorcet dimension $k + 1$.
\end{theorem}

\begin{proof}
Consider the following matching instance for $k \in \mathbb{N}$.
Let~\mbox{$G = (A \cup O, E)$} be the complete bipartite graph with $k^2 + k + 1$ nodes on each side.
The set of objects~$O$ is partitioned into three sets: $\{o^*\}$, $O_1$, and  $O_2$, where $|O_1| = k$ and $O_2 = k^2$.
All $k^2 + k + 1$ agents have identical preferences, namely they prefer $o^*$ over any object from $O_2$ and are indifferent about any other comparison.
Consider any set of at most $k$ matchings $\mathcal{M}$ in $G$.
Let $A^* \subseteq A$ be the set of agents assigned to $o^*$ in at least one matching in $\mathcal{M}$.
Note that $|A^*| \leq k$.
Furthermore, there is at least one agent $a' \in A$ that is not assigned to any object in $O_1 \cup \{o^*\}$ in any of the matchings in $\mathcal{M}$, because $k \cdot |O_1 \cup \{o^*\}| = k^2 + k < |A|$.
Construct the matching $N$ by assigning $a'$ to $o^*$, assigning all agents in $A^*$ to some object in $O_1$, and assigning all remaining agents arbitrarily to the remaining objects.
Note that $a^*$ prefers $N$ over any matching in $\mathcal{M}$ (as in any of these matchings, they are either assigned to an object from $O_2$ or not at all) and no agent prefers any matching in $\mathcal{M}$ over $N$ (as all agents in $A^*$ are assigned an object in $O_1$ and all other agents are not unassigned).
Hence, $\phi(\mathcal{M}, N) < 0$, i.e., $\mathcal{M}$ is not popular.
Hence the weak Condorcet dimension in the constructed instance is at least~$k + 1 > \sqrt{|A|} = \sqrt{k^2 + k + 1}$.
\end{proof}

\subsection{An Algorithmic Upper Bound on the Weak Condorcet Dimension for Matchings on Partial-order Preferences}
\label{sec:partial-order-matchings-ub}

We now describe an algorithm that computes a popular set of at most $\lceil \sqrt{2|A|} \rceil$ matchings for any given matching instance with partial-order preferences, thus matching the lower bound given above up to a factor of $\sqrt{2}$.

\begin{theorem}\label{thm:partial-order-ub}
	There is an algorithm that, given a bipartite graph $G = (A \cup O, E)$ and partial order preferences $\prec_a$ for each agent $a \in A$, computes a popular set of matchings in~$G$ containing at most $\lceil\sqrt{2|A|}\rceil$ matchings.
	In particular, the weak Condorcet dimension of any matching instance under partial order preferences is at most $\lceil\sqrt{2|A|}\rceil$.
\end{theorem}

\begin{proof}
    In the following, we let $\Gamma_G(X)$ for $X \subseteq A$ denote the set of objects adjacent to at least one agent in $X$ in the graph $G$. 
    Let $k := \lceil\sqrt{2|A|}\rceil$.
    We construct a set of $k$-matchings in $G$ as follows.
    \begin{enumerate}
        \item Let $i := 0$, $A_0 := A$, $O_0 := O$, and $E_0 := \{(a, o) \in E : o \not\prec_a o' \text{ for all } o' \in O\}$.
        \item As long as there is no $k$-matching $\bar{M}$ in $G_i = (A_i \cup O_i, E_i)$ covering all agents in~$A_i$, repeat the following:
        \begin{itemize} 
            \item By Hall's theorem, there exists $X_i \subseteq A_i$ with $k \cdot |\Gamma_{G_i}(X_i)| < |X_i|$.
            Note that~$|\Gamma_{G_i}(a)| < k$ for all $a \in X_i$, because otherwise $k \cdot |\Gamma_{G_i}(X_i)| \geq |A| \geq |X_i|$.
            Let $X'_i \subseteq X_i$ with $|X'_i| = k$ and let $\bar{M}_i$ be the $k$-matching that assigns each agent $a \in X'_i$ to each object in $\Gamma_{G_i}(a)$. 
            \item Set $i := i + 1$ and update the sets
            \begin{align*}
                A_{i} & := A_{i-1} \setminus X'_{i-1},\\
                O_{i} & := O_{i-1} \setminus \Gamma_{G_{i-1}}(X_{i-1}),\\
                E_{i} & := \{(a, b) \in E : a \in A_{i}, o \in O_{i}, o \not\prec_a o' \text{ for all } o' \in O_{i}\}.
            \end{align*}
        \end{itemize}
        \item Let $\ell$ be the value of $i$ when the while loop in step (2) terminates and let $\bar{M}_{\ell}$ be the $k$-matching that satisfies the termination criterion of the while loop.
        Let $\bar{M} := \bar{M}_0 \cup \dots \cup \bar{M}_{\ell}$.
    \end{enumerate}

     Note that the sets of agents and objects covered in each of the $k$-matchings $M_j$ for $j \in [\ell]$ are disjoint and thus $\bar{M}$ is itself a $k$-matching.
     Thus, let $\mathcal{M}$ be a decomposition of $\bar{M}$ into $k$ disjoint matchings.
     We claim that $\mathcal{M}$ is a popular set of matchings.
     For this, consider any matching $M'$ in $G = (A \cup O, E)$.
    
    First, let $A_+ := \{a \in A : M'(a) \succ_a M(a) \text{ for all } M \in \mathcal{M}\}$.
    Consider any $a \in A_+$ and let $j'$ be the unique index such that $a$ is assigned by $\bar{M}_{j'}$.
    Note that $M'(a) \succ_a b$ for all $b \in \bar{M}_{j'}(a)$. 
    Because the objects in $\bar{M}_{j'}(a) \subseteq O_{j'}$ are undominated in $O_{j'}$, the object $M'(a)$ must be in $B_{j''}$ for some $j'' < j' \leq \ell$.
    In particular, each agent in $A_+$ must be assigned to an object in $\bar{O} := O \setminus O_{\ell}$.
    Hence $|A_+| \leq |\bar{O}|$.
    
    Secondly, let $A_- := \{a \in A : M'(a) \prec_a M(a) \text{ for some } M \in \mathcal{M}\}$.
    We show that for any $j \in \{0, \dots, \ell - 1\}$ and any agent $a \in X'_j$ with $M'(a) \notin \bar{B}$ we have $a \in A_-$.
    Indeed, note that  any such agent $a$ is assigned to all objects that are undominated for $a$ in $G_j$ and therefore for each $o \in O_{j+1}$ there is some $M \in \mathcal{M}$ with $M(a) \succ_a o$.
    Hence, unless $M'(a) \in O \setminus O_{j+1} \subseteq \bar{O}$, we have $M'(a) \prec_a M(a)$ for some $M \in \mathcal{M}$.
    Therefore $|A_-| \geq |\bigcup_{j=0}^{i-1} X'_j| - |\bar{O}| = \ell \cdot k - |\bar{O}|$. 
    
    Observe that $|O_{\ell}| \geq |O | - \frac{\ell|A|}{k}$, which follows by induction with $O_0 = O$ and $O_{j+1} = O_{j} \setminus \Gamma_{G_j}(X_j)$, using $|\Gamma_{G_j}(X_j)| < \frac{X_j}{k} \leq \frac{|A|}{k}$ for $j \in \{0, \dots, \ell-1\}$. From this we conclude $|\bar{O}| = |O \setminus O_{\ell}| \leq \frac{\ell|A|}{k} \leq \sqrt{\frac{|A|}{2}}$ by choice of $k \geq \sqrt{2|A|}$. Therefore
    \begin{align*}
        |A_+| - |A_-| & \leq |\bar{O}| - (\ell \cdot k - |\bar{O}|) = 
        2|\bar{O}| - \ell k\\
        & \leq 2 \ell \sqrt{|A|/2} - \ell \sqrt{2|A|} = 0.
    \end{align*}
    We conclude that $\mathcal{M}$ is a popular set of matchings.
\end{proof}

\subsection{The Condorcet Dimension for Matroid-constrained Matchings under Partial-order Preferences}
\label{sec:partial-order-matchings-matroids}

A trivial upper bound on the weak Condorcet dimension of any election with partial-order preferences is the number of agents $|A|$: simply choose for every agent $a$ an alternative $X_a$ that is undominated under $a$'s preferences (e.g., in the case of matchings: a matching $M_a$ that matches $a$ to an undominated object); then $\{X_a : a \in A\}$ is a popular set of alternatives.
We now show that in the presence of partial-order preferences, the weak Condorcet dimension for matroid-constrained matchings can, in fact, 
be this trivial upper bound of $|A|$.

\begin{theorem}
    For every $k \in \mathbb{N}$ with $k \geq 2$, there is a matroid-constrained matching instance with~$k$ agents having partial-order preferences and weak Condorcet dimension~$k$.
\end{theorem}
\begin{proof}
Let $k \in \mathbb{N}$ and let $m := k(k+1)$. 
We construct a matroid-constrained matching instance with complete bipartite graph $G = (A \cup O, E)$, matroid $(O, \mathcal{I})$, and  partial-order preferences for the agents as follows.
The set of objects is $O = O' \cup O''$, where
\begin{align*}
	O' := \{o_1, \dots, o_m\} \text{ and } O'' = \left\{o_{S} :  S \in \binom{O'}{k}\right\}.
\end{align*}
The matroid on $O$ is given by $\mathcal{I} := \{I \subseteq O : |I \cap O''| \leq 1\}$.
The set of agents $A$ consists of $k+1$ agents with identical partial order preferences, where $o_S \succ_a o$ for each $S \in \binom{O'}{k}$ and $o \in S$; 
the agents are indifferent between any pair of objects that is not of this form.

Now consider any set of $k$ feasible matchings $\mathcal{M} = \{M_1, \dots, M_k\}$ in the constructed instance.
We show that $\mathcal{M}$ is not popular, and thus the weak Condorcet dimension of the constructed instance is at least $k + 1 = |A|$.
Note that for each $i \in [k]$, because $M_i(A) \in \mathcal{I}$, there is at most one agent $a \in A$ with $M_i(a) \in O''$.
Because there are $k + 1$ agents, there must thus be at least one agent $a_0$ with $M_i(a_0) \notin O''$.
Let $T_0 := \{M_i(a') : i \in [k]\} \subseteq O'$.
Note that $|T_0| \leq k$ and let $S_0$ be the union of $|T_|$ with $k - |T_0|$ additional distinct objects from $O' \setminus T_0$.
For $a \in A \setminus \{a_0\}$ and $i \in [k]$, define $O^{\prec}_{a,i} := \{o \in O' : M_i(a) \succ o\}$.
Note that $|O^{\prec}_{a,i}| = k$ if $M_i(a) \in O''$ and that $|O^{\prec}_{a,i}| = 0$ otherwise.
Let $\bar{O}' := O' \setminus \bigcup_{a \in A \setminus \{a_0\}, i \in [k]} O^{\prec}_{a,i}$.
Note that $|\bar{O}'| \geq |O'| - k^2 = m - k^2 = k$,
because for each $i \in [k]$ there is at most one $a \in A \setminus \{a_0\}$ with $M_i(a) \in O''$. 
Let $M'$ be an arbitrary perfect matching between the agents in $A \setminus \{a_0\}$ and the objects in $\bar{O}'$, complemented by $M'(a_0) := o_{S_0}$.
Note that $M'$ is a feasible matching such that $M'(a_0) = o_{S_0} \succ M_i(a)$ for all $i \in [k]$ by choice of $S_0$, and that $M'(a) \not\prec M_i(a)$ for each $a \in A \setminus \{a_0\}$ because $M'(a) \notin O^{\prec}_{a,i}$.
Thus $\mu(\mathcal{M}, M') = -1$, i.e., $\mathcal{M}$ is not popular.
\end{proof}

\section{Hardness Results for Matchings under Partial-order Preferences}
\label{sec:partial-order-hardness}

In this section, we consider the computational complexity of two problems for matching instances with partial-order preferences.
First, in \cref{sec:hardness-condorcet}, we consider the problem of determining the Condorcet dimension of a given matching instance under partial order preferences (which may be significantly better than the $\lceil\sqrt{2|A|}\rceil$ worst case established in \cref{sec:partial-orders}).
Second, in \cref{sec:hardness-pareto}, given the fact that Pareto-optimal solutions might fail to exist under partial-order preferences, we consider the question whether we can compute a Pareto-optimal matching for a given matching instance with partial-order preferences, or decide that no such matching exists.
Unfortunately, in both cases the problem turns out to be NP-complete.

\subsection{Determining the Condorcet Dimension is Hard for Matchings under Partial-order Preferences}
\label{sec:hardness-condorcet}

We consider the problem of determining the Condorcet dimension of a given matching instance under partial order preferences.
It turns out that this problem is NP-hard, even when we want to check whether the Condorcet dimension is bounded by a constant value $k$.

\begin{restatable}{theorem}{restateThmNPCondorcet}\label{thm:NP-condorcet}
    The following problem is NP-complete for any constant $k \in \mathbb{N}$,~$k > 1$: Given a matching instance $(G, \succ)$ with partial-order preferences, decide whether there is a popular set of matchings in $G$ of size at most $k$.
\end{restatable}

The theorem is proven via reduction from \textsc{$\ell$-Dimensional Matching}.
The complete proof is discussed in \cref{app:NP-condorcet}.

\subsection{Finding a Pareto-optimal Matching is Hard under Partial-order Preferences}
\label{sec:hardness-pareto}

We now show that determining whether a matching instance under partial-order preferences admits a Pareto-optimal solution is NP-complete.

\begin{theorem}\label{thm:NP-pareto}
    The following problem is NP-complete: Given a matching instance~\mbox{$(G, \succ)$} with partial-order preferences, decide whether there is a Pareto-optimal matching in $G$.
\end{theorem}

We split the proof into four parts. First we show containment in NP, then we describe the construction of our reduction (from \textsc{Vertex Cover}), then we show how a Pareto-optimal matching in the resulting instance can be constructed from a vertex cover, and vice versa.
Together, these four steps comprise the entire proof of the theorem.

\subsubsection*{Containment in NP}
    Note that the problem is in NP, as for a given matching $M$ in $G = (A \cup O, E)$ we can check whether it is Pareto-optimal by trying to construct a matching that dominates $M$ as follows: First, guess an agent $\hat{a} \in A$ and let $$E_{\hat{a}} := \{(\hat{a}, o) \in E : o \succ_{\hat{a}} M(a)\} \cup \{(a, o) \in E : a \neq \hat{a}, M(a) \not\succ_a o\}.$$
    Then check if $G_{\hat{a}} = (A \cup O, E_{\hat{a}})$ contains a matching.
    If this is the case for some choice of $\hat{a} \in A$, then we found a matching dominating $M$; if this is not the case for any $\hat{a} \in A$, then $M$ is Pareto-optimal.

\subsubsection*{Construction: Reduction from \textsc{Vertex Cover}}
    We establish NP-hardness via reduction from \textsc{Vertex Cover}: Given a graph $G' = (V', E')$ and an integer $\ell$, does there exist a subset of vertices $U \subseteq V'$ such that $U \cap e \neq \emptyset$ for all $e \in E'$? Given an instance of \textsc{Vertex Cover}, we construct a matching instance $G = (A \cup O, E)$ with partial order preferences as follows.

    The set of objects $O$ contains the following:
    \begin{itemize}
        \item An object $o_v$ for each $v \in V'$.
        \item Objects $o^v_{e}, o^w_{e}, o^1_e, o^2_e, o^3_e$ for each $e = \{v, w\} \in E$.
        \item A set of $|V|-\ell$ identical objects $\bar{O}$.
    \end{itemize}
    The set of agents $A$ contains the following:
    \begin{itemize}
        \item An agent $a_v$ for each $v \in V$. This agent is adjacent to $o_v$ and to each $\bar{o} \in \bar{O}$. The agent prefers each $\bar{o} \in \bar{O}$ to $o_v$.
        \item Agents $a_e, a^v_e, a^w_e, a^1_e, a^2_e, a^3_e$ for each $e = \{v, w\} \in E'$. 
        \begin{itemize}
            \item Agent $a_e$ is incident to $o^v_e$, $o^w_e$, $o_v$, and $o_w$. The agent prefers $o^v_e$ over $o_v$ and prefers $o^w_e$ over $o_w$.
            \item Agent $a^v_e$ is adjacent to $o^v_e$ and $o^1_e$, preferring $o^v_e$ over $o^1_e$.
            \item Agent $a^w_e$ is adjacent to $o^w_e$ and $o^1_e$, preferring $o^w_e$ over $o^1_e$.
            \item Agents $a^1_e$, $a^2_e$, $a^3_e$ each are adjacent to $o^1_e$, $o^2_e$, and $o^3_e$, preferring $o^1_e$ over $o^2_e$.
        \end{itemize}
    \end{itemize}
    \begin{figure}
        \centering
        \includegraphics[width=\linewidth]{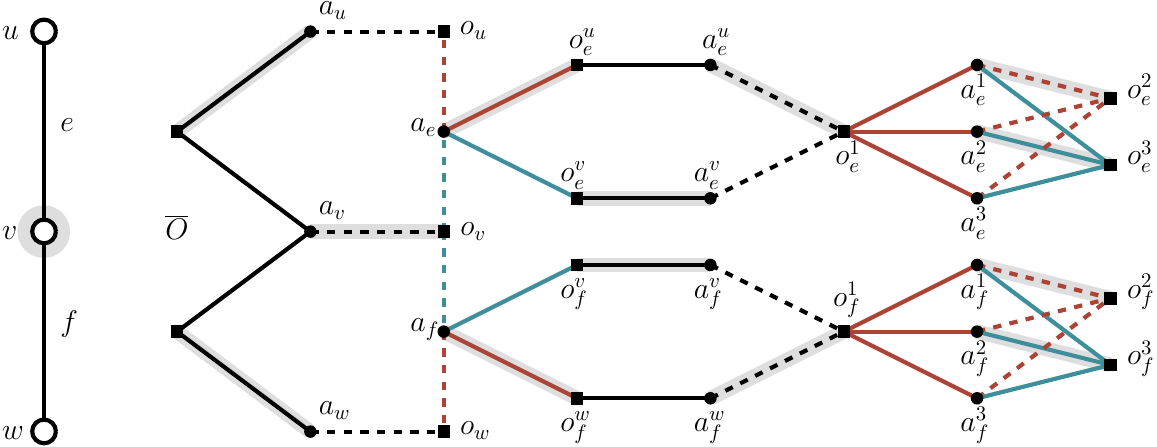}
        \caption{The construction used to prove \cref{thm:NP-pareto}. 
        The graph on nodes $V = \{u, v, w\}$ on the left depicts a \textsc{Vertex Cover} instance, for which $\{v\}$ is a vertex cover of size $\ell = 1$. The resulting matching instance is depicted on the right, with $\bar{O}$ containing $|V| - \ell = 2$ nodes. 
        Any agent prefers a solid edge over a dashed edges of the same color, but is indifferent for all other pairs (in particular edges of different colors). 
        The edges marked with grey background form a Pareto-optimal matching.%
        }
        \label{fig:reduction}
    \end{figure}

    See \cref{fig:reduction} for an illustration of the construction.
    In \cref{app:NP-Pareto}, we show that the constructed instance admits a Pareto-optimal matching if and only if there is a vertex cover of size at most $\ell$.

\section{The Condorcet Dimension for Arborescences}

In this section, we consider the setting where the alternatives are not matchings but arborescences.
Recall, that an arborescence instance is given by an arborescence instance $(D = (A \cup \{r\}, E), \succ)$ consists of a digraph $D$ with root $r$ and partial order preferences of the agents over their incoming arcs.
We are interested in finding a (weak) Condorcet-winning set among those $r$-arborescences.
The following theorem implies that under partial order preferences, two arborescences suffice to construct a weak Condorcet-winning set.
For strict rankings, the two arborescences form a Condorcet-winning set (unless there exists a top-choice arborescence, which yields a Condorcet-winning  set of size $1$).

\begin{theorem}
    There exist two (possibly identical) $r$-arborescences $T_1, T_2 \subseteq E$ such that $T_1 \cup T_2$ contains a $\succ_a$-maximal arc from $\delta^-(a)$ for every $a \in A$. Such arborescences can be computed in polynomial time.
\end{theorem}

\begin{proof}
    We first note that without loss of generality, we can assume that every arc $e \in E$ is contained in some $r$-arborescence, as we can remove any arc which does not fulfill this property without changing the set of arborescences.
    In particular, this implies that for every $(a, b) \in E$ there is an $r$-$b$-path in $D$ whose last arc is $(a, b)$.
    
    For each agent $a \in A$, we fix a $\succ_a$-maximal arc $e_a \in \delta_{D}^-(a)$.
    We construct a digraph $D' = (V' \cup \{r\}, E')$ from $D$ as follows: 
    For each agent $a \in A$, we include $a$ as a node in $V'$. 
    For each agent $a \in A$, with $|\delta_{D}^-(a)| > 2$, we further include a node $v_{a}$ in $V'$.
    We include all arcs $e_a$ for $a \in A$ in $E'$.
    Furthermore, if $|\delta_{D}^-(a)| = 1$ for some $a \in A$, we include a parallel copy $e'_a$ of $e_a$ in $E'$.
    Finally, if $|\delta_{D}^-(a)| > 2$ for some $a \in A$, we include the arc $(v_a, a)$ and for each $(\bar{a}, a) \in \delta_{D}^-(a) \setminus \{e_a\}$ the arc $(\bar{a}, v_a)$ in $E'$.
    Note that $|\delta_{D'}^-(v)| \geq 2$ for all $v \in V'$, with equality if $v = a$ for some $a \in A$.
    Furthermore, any $r$-arborescence $T'$ in $D'$ induces an arborescence $\phi(T')$ in $D$ by replacing each path of the form $\{(\bar{a}, v_a), (v_a, a)\}$ by the arc $(\bar{a}, a)$, and each $e'_a$ by $e_a$.
    Note that if $e_a \in T'$ for some $a \in A$ then $e_a \in \phi(T')$.
    Thus, to prove the theorem, it suffices to show that there are two $r$-arborescences in $D'$ whose union contains $e_a$ for all $a \in A$.
    To see that this is indeed the case, we show that $D'$ contains two disjoint $r$-arborescences.
    Because each agent $a \in A$ has exactly $2$ incoming arcs in $D'$, one of which is $e_a$, one of the two arborescences must contain $e_a$.
    
    We now show that $|\delta_{D'}^-(S)| \geq 2$ for all $S \subseteq V'$ with $S \neq \emptyset$, which implies the existence of two disjoint arborescences in $D'$ by \cite{edmonds1973edge}.
    We first observe that for every arc $(u, v) \in E'$ there is an $r$-$v$-path in $D'$ whose last arc is $(u, v)$, a property $D'$ directly inherits from $D$ by construction.
    Let $S \subseteq V'$ with $S \neq \emptyset$. 
    Note that $\delta^-_{D'}(v) \neq \emptyset$ because $r \notin S$ and because there is an $r$-$v$-path in $D'$ for every $v \in V$.
    Thus let $(w, v) \in \delta^-(S)_{D'}$ and let $(u, v) \in \delta^-_{D'}(v) \setminus \{(w, v)\}$, which exists because $|\delta^-_{D'}(v)| \geq 2$.
    Note that an $r$-$v$-path containing $(u, v)$ as its last arc cannot contain $(w, v)$ and thus there must by an arc different from $(w, v)$ that enters $S$.
    We conclude that $|\delta_{D'}^-(S)| \geq 2$, which completes the proof.
\end{proof}

\section{Discussion}\label{sec:discussion}

In this work, we apply the concept of Condorcet-winning sets, which has recently received increasing attention \cite{CLR+25a,song2026few}, to elections induced by (graph) structures under preferences. While in general social choice it remains open whether the Condorcet dimension can be $3$, $4$, or $5$, our setting allows us to exploit the underlying combinatorial structure and obtain (asymptotically) tight bounds on the (weak) Condorcet dimension across our three models.

Because this paper (together with \cite{connor2025popular}) is, to our knowledge, the first to study the (weak) Condorcet dimension when agents may have weak rankings over alternatives, a natural next step is to move beyond combinatorial elections to standard ordinal voting and derive bounds on the (weak) Condorcet dimension when agents have weak rankings.

For combinatorial elections, an interesting direction is the matroid-constrained model with the additional requirement that the selected objects form a basis (see \Cref{sec:prelim}), as studied by \citet{KMS+25a}. This variant generalizes both the arborescence and assignment \cite{kavitha2022popular} models. However, we can rule out the extension of our main result: in the assignment model, even with strict rankings, Pareto-optimal sets of size $2$ need not be Condorcet-winning (see \Cref{sec:appassignment}).

\section*{Acknowledgements}
This work was initiated at the Dagstuhl Seminar on Scheduling (25121). We thank Adrian Vetta, whose talk inspired this work, as well as the seminar organizers, Claire Mathieu, Nicole Megow, Benjamin Moseley, and Frits Spieksma. 
Telikepalli Kavitha was supported by the Department of Atomic Energy, Government  of India, under project no. RTI4014.
Jannik Matuschke was supported by the special research fund of KU Leuven (project C14/22/026).
Ulrike Schmidt-Kraepelin was supported by the Dutch Research Council (NWO) under project number VI.Veni.232.254.

\bibliographystyle{ACM-Reference-Format}
\bibliography{references}

\appendix

\section{Missing Proofs from \cref{sec:pareto-implies-condorcet-matroids}}
\label{app:matroids}

\subsection{Proof of \cref{lem:computing-pareto}}
\label{app:matroids-computing}

\restateLemComputingPareto*

\begin{proof}
    For $a \in A$, let $O_a$ denote the set of objects adjacent to $a$ in $G$.
    Without loss of generality, we can assume that $O_a \cap O_{a'} = \emptyset$ for distinct agents $a, a' \in A$: If there is $o \in O_a \cap O_{a'}$, we introduce a new element $o'$ to $O$, replace the edge $(a', o)$ by $(a', o')$ in $E$, and make $o'$ a parallel element to $o$ in $\I$, i.e., no set containing both $o, o'$ is in $\I$, and for $I \subseteq O \setminus \{o, o'\}$ it holds that $I \cup \{o'\} \in \I$ if and only if $I \cup \{o\} \in \I$.
    
    We define a maximum-weight matroid intersection problem as follows.
    Let \begin{align*}
        \mathcal{I}_A & := \{I \subseteq O : |I \cap O_a| \leq 1 \text{ for all } a \in A\} \text{ and }\\
        \mathcal{I}_O & := \{I \subseteq O : I = I' \cup I'' \text{ for some } I, I'' \in \mathcal{I}\}.
    \end{align*}
    Note that both $\mathcal{I}_A$ and $\mathcal{I}_O$ are matroids: the first is a partition matroid, the second is a matroid union (of $\mathcal{I}$ with itself).
    Moreover, for $a \in A$ and $o \in O_a$ we define the weight~$w_{o} := |O| - |\{o' \in O_a : o' \succ_a o\}|$.
    Note that, because $\prec_a$ is a weak ranking, for~$o, \hat{o} \in O_a$ we have $w_o > w_{\hat{o}}$ if and only if $o \succ_a \hat{o}$.
    
    We apply \citeauthor{frank1981weighted}'s~[\citeyear{frank1981weighted}] weighted matroid intersection algorithm to find a set $S \subseteq O$ with $S \in \mathcal{I}_A \cap \mathcal{I}_O$ maximizing $\sum_{o \in S} w_o$.
    Let $I', I'' \in \mathcal{I}$ such that $S = I' \cup I''$ (which exist because $S \in \mathcal{I}_A$). 
    Note that we can assume $I' \cap I'' = \emptyset$ without loss of generality.
    Let $M' := \{(a, o) : a \in A, o \in I' \cap O_a\}$ and $M'' := \{(a, o) : a \in A, o \in I' \cap O_a\}$.
    Note that both $M'$ and $M''$ are $\mathcal{I}\!$-constrained matchings by construction.
    
    Finally, we show that $\mathcal{M} := \{M', M''\}$ is Pareto-optimal.
    By contradiction assume that there is a set of two $\mathcal{I}\!$-constrained matchings $\hat{\mathcal{M}} = \{N', N''\}$ that Pareto-dominates $\mathcal{M}$.
    Without loss of generality, we can assume that no agent $a \in A$ is matched in both $N'$ and $N''$ (because preferences are weak rankings) implying that $\hat{S} := N'(A) \cup N''(A) \in \mathcal{I}_A \cap \mathcal{I}_O$.
    For $a \in A$ let $S(a)$ and $\hat{S}(a)$, respectively, be the unique elements of $O_a$ such that $S(a) \in S$ and $\hat{S}(a) \in \hat{S}$, respectively (if no such element exists, then $S(a) = \emptyset$, $\hat{S}(a) \in \hat{S}$, respectively).
    Because $\hat{\mathcal{M}}$ Pareto-dominates $\mathcal{M}$ it holds that that $S(a) \not\succ_a \hat{S}(a)$ which implies $|\{o \in O_a : o \succ_a \hat{S}(a)\}| \leq |\{o \in O_a : o \succ_a S(a)\}|$ (here we crucially use that $\prec_a$ is a weak ranking) and therefore $w_{\hat{S}(a)} \geq w_{S(a)}$.
    Moreover, there is at least one agent $\hat{a} \in A$ such that $\hat{S}(\hat{a}) \succ_{\hat{a}} S(\hat{a})$ and therefore the above inequality is strict.
    We obtain $\sum_{o \in S} w_o = \sum_{a \in A : S(a) \neq \emptyset} w_{S(a)} > \sum_{a \in A : \hat{S}(a) \neq \emptyset} w_{\hat{S}(a)} = \sum_{o \in \hat{S}} w_o$, contradicting the fact that $S$ is weight-maximizing.
\end{proof}

\subsection{Proofs of \cref{lem:exchange-cycle,cor:exchange-path}}
\label{app:matroids-exchange}

To prove \cref{lem:exchange-cycle}, we use the following result by \citet{frank1981weighted}, restated here in the version given by \citet[Lemma~13.27]{korte2018combinatorial}.

\begin{lemma}[\cite{frank1981weighted}]\label{lem:frank-exchange}
    Let $(O, \I)$ be a matroid and let $S \in \I$.
    Let $x_1, \dots, x_k \in \I$ and $y_1, \dots, y_k \in O \setminus S$ with $x_i \in C(S, y_i)$ for $i \in [k]$ and $x_j \notin C(S, y_i)$ for $i, j \in [k]$ with $j < i$.
    Then $S \setminus \{x_1, \dots, x_k\} \cup \{y_1, \dots, y_k\} \in \I$.
\end{lemma}

\restateLemExchangeCycle*

\begin{proof}
    Let $C'$ be a cycle in $D_{N \rightarrow M}$ on the nodes in $A(C)$ that  contains $a'$ and has a minimum number of nodes among all such cycles.
    Let $A' = \{a_1, \dots, a_k\}$ be the nodes on $C'$, labeled in the order they appear when traversing the cycle in direction of its arcs, starting with $a_1 := a'$.
    Observe that $a_{i+1} \in C(M(A), N(a_i))$ for $i \in [k]$ because $(a_i, a_{i+1}) \in C' \subseteq E_{N \rightarrow M}$.
    Note moreover that, by choice of $C'$, there is no edge $(a_i, a_j)$ in $D_{N \rightarrow M}$ for any $i, j \in [k]$ with $j > i + 1$, i.e., $M(a_j) \notin C(M(A), N(a_i))$ for any $i, j \in [k]$ with $j > i + 1$.
    We may thus apply \cref{lem:frank-exchange} to $M(A)$ with $x_i = M(a_{i+1})$ (where $x_{k+1} := x_1$) and $y_i = N(a_i)$ for $i \in [k]$.
    From this we obtain that $M'(A) = M(A) \setminus \{M(a_1), \dots, M(a_k)\} \cup \{N(a_1), \dots, N(a_k)\}$ is a basis of $\mathcal{I}$, where~$M'$ is the matching described in the statement of the lemma. 
    Thus, $M'$ is indeed an~{\Iconstr} matching.
\end{proof}

\restateCorExchangePath*

\begin{proof}
    Let $N'$ be the {\Iconstr} matching with $N'(a) = N(a)$ for $a \in A \setminus \{b'\}$ and $N'(b') = \emptyset_{b'}$.
    Note that $E_{N \rightarrow M} \subseteq E_{N' \rightarrow M}$, because the only arcs whose presence depends on the value of $N(b')$ or $N'(b')$, respectively, are the out-arcs of $b'$, and $E_{N' \rightarrow M}$ contains the arc $(b', a)$ for every $a \in A$ because $C(M(A), \emptyset_b) = M(A)$.
    In particular, $P$ is an $a'$-$b'$-path in $D_{N' \rightarrow M}$ and can be extended with the arc $(b', a')$ to a cycle $C$ in $D_{N' \rightarrow M}$.
    Applying \cref{lem:exchange-cycle} to $N'$, $M$, and $C$ yields the corollary.
\end{proof}

\subsection{Proof of \cref{lem:B-indegree}}
\label{app:matroids-branching-indegree}

\restateLemBindegree*

\begin{proof}
        Let $a' \in A$.
        Recall that for each $M \in \M$, the paths $P^M_a$ for $a \in A^M_+$ are node-disjoint, and there is at most one such path containing $a'$ for each $M \in \M$.
        Now assume by contradiction there are $M_1, M_2 \in \mathcal{M}$ with $M_1 \neq M_2$ such that $a'$ is contained both in~$P^{M_1}_{a_1}$ and $P^{M_2}_{a_2}$.
        Without loss of generality, we can assume that~$M_2(a') \not\prec_{a'} M_1(a')$~(if this does not hold, we can swap $M_1$ and $M_2$; note that $M_2(a') \prec_{a'} M_1(a')$ and $M_1(a') \prec_{a'} M_2(a')$ cannot hold simultaneously).

        Because $a'$ is contained in $P^{M_1}_{a_1}$, there is an $a_1$-$a'$-path $P_1 \subseteq P^{M_1}_{a_1}$.
        We apply \cref{cor:exchange-path} to the $a_1$-$a'$-path $P_1$ in $D_{M_1 \rightarrow N}$ to obtain $A' \subseteq A(P_1)$ with $a_1, a' \in A'$ and an {\Iconstr} matching $M'$ with $M'(a) = N(a)$ for all $a \in A' \setminus \{a'\}$ and $M'(a) = M_1(a)$ for all $a \in A \setminus A'$.
        Let $\mathcal{M}' := \mathcal{M} \setminus \{M_1\} \cup \{M'\}$.
        We show that $\mathcal{M}'$ dominates $\mathcal{M}$, leading to a contradiction.
        
        Note that $M'(a_1) = N(a_1) \succ_{a_1} M(a_1)$ for all $M \in \mathcal{M}$ because $a_1 \in A_+$.
        We further show that for all $a \in A$ and all $M \in \mathcal{M}$ there is some $\bar{M} \in \mathcal{M}'$ with $M(a) \not\succ_a \bar{M}(a)$.
        So let $a \in A$ and $M \in \mathcal{M}$.
        If $M \neq M_1$ we are done because then $M \in \mathcal{M} \setminus \{M_1\} \subseteq \mathcal{M}'$.
        To show that the statement also holds for $M = M_1$, we distinguish three cases:
        \begin{itemize}
            \item If $a = a'$ then $M_1(a') \not\succ_{a'} M_2(a') \in \mathcal{M}'$.
            \item If $a \in A \setminus A'$, then $M_1(a) = M'(a)$ and hence $M_1(a) \not\succ_a M'(a)$.
            \item If $a \in A' \setminus \{a'\}$ then $a$ is an internal node of the $a_1$-$a'$-path $P_1$, which is a subpath of $P^M_{a_1}$. 
            The internal nodes of the latter path are not in $A_-$ and thus $a' \notin A_-$. Hence $M_1(a) \not\succ_a N(a) = M'(a)$.
        \end{itemize}
        Thus $\mathcal{M}'$ dominates $\mathcal{M}$, a contradiction. 
    \end{proof}

\subsection{Proof of \cref{lem:B-nocycles}}
\label{app:matroids-branching-no-cycle}

We provide the proof of \cref{lem:B-nocycles}.
The proof is based on an application of the following generalization of \cref{cor:exchange-path}, which we prove first.

\begin{corollary}\label{cor:exchange-path-multiple}
    Let $M, N$ be two {\Iconstr} matchings in $G$ and let $P_1, \dots, P_{\ell}$ be node-disjoint paths in $D_{N \rightarrow M}$ with $P_i$ for $i \in [\ell]$ being an $a_i$-$b_i$-path for $i \in [\ell]$.
    Assume further that there is no arc $(a, b) \in E_{N \rightarrow M}$ for any $a \in A(P_i)$ and $b \in A(P_j)$ with $i > j$.
    Then there is $A' \subseteq \bigcup_{i \in [\ell]} A(P_i)$ with $a_1, \dots, a_{\ell}, b_1, \dots, b_{\ell} \in A'$ such that $M'$ defined by 
    $$M'(a) := \begin{cases} 
        M(a) & \text{ for } a \in A \setminus A',\\
        N(a) & \text{ for } a \in A' \setminus \{b_1, \dots, b_{\ell}\},\\
        \emptyset_{b_i} & \text{ for } i \in [\ell]
    \end{cases}$$
        is an {\Iconstr} matching in $G$.
\end{corollary}

\begin{proof}
Let $M_1 := M$.
We construct the {\Iconstr} matching $M_{k+1}$ for $k \in \{1, \dots, {\ell}\}$ inductively by applying \cref{cor:exchange-path} to $M_k$ and the path $P_k$ and letting $M_{k+1}$ be the resulting $M'$.
To show that this is possible, we prove that all arcs of the paths $P_k, \dots, P_{\ell}$ are contained $D_{N \rightarrow M_k}$ and that $D_{N \rightarrow M_k}$ contains no arc $(a, b)$ with $a \in A(P_i)$ and $b \in A(P_j)$ for $i, j \in \{k, \dots, \ell\}$ with $i > j$.

The base case $k=1$ follows trivially from the premise of the lemma.
Thus consider the induction step for $k > 1$:
By induction hypothesis it holds that $D_{N \rightarrow M_k}$ contains the paths $P_k, \dots, P_{\ell}$ but no arc $(a, b)$ with $a \in A(P_i)$ and $b \in A(P_j)$ for $i, j \in \{k, \dots, \ell\}$ with $i > j$.
We construct $M_{k+1}$ from $M_k$ by applying \cref{cor:exchange-path} to $M_k$ with the path $P_k$, resulting in $A'_k \subseteq A(P_k)$ with $a_k, b_k \in A'_k$ and an {\Iconstr} matching $M_{k+1}$ with $M_{k+1} = N(a)$ for $a \in A_k' \setminus \{b_k\}$, $M_{k+1}(b_k) = \emptyset_{b_{\ell}}$, and $M_{k+1}(a) = M_k(a)$ for $a \in A \setminus A'_k$.
Let $A_{k+1} := \bigcup_{j=k+1}^{\ell} A(P_j)$.
Note that $M_{k+1}(a) = M_{k}(a)$ for $a \in A_{k+1} \subseteq A \setminus A(P_k) \subseteq A \setminus A'_k$.
Note further that the induction hypothesis implies the absence of arcs $(a, b)$ with $a \in A_{k+1}$ to $b \in A(P_k)$ in $D_{N \rightarrow M_k}$ and hence $M_k(b) \notin C(M_k(A), N(a))$ for any $a \in A_{k+1}$ and any $b \in A(P_k)$.
In particular, if $M_k(b) \in C(M_k(A), N(a))$ for $a \in A_{k+1}$, then $M_{k+1}(b) = M_{k}(b)$, from which we conclude $C(M_k(A), N(a)) \subseteq M_{k+1}(A) \cup \{N(a)\}$. 
As $M_{k+1}(A) \cup \{N(a)\}$ only contains a unique circuit, this implies $C(M_{k+1}(A), N(a)) = C(M_k(A), N(a))$ for all $a \in A_{k+1}$.
We conclude that for all $a \in A_{k+1}$ we have both $M_{k+1}(a) = M_{k}(a)$ and  $C(M_{k+1}(A), N(a)) = C(M_k(A), N(a))$, which implies that for $a, b \in A_{k+1}$ we have $(a, b) \in E_{N \rightarrow M_{k+1}}$ if and only if $(a, b) \in E_{N \rightarrow M_{k}}$.
Thus, we have established that the induction hypothesis continues to hold for $k+1$.

In particular, this process yields $A' := \bigcup_{k=1}^{\ell} A'_k$ and an {\Iconstr} matching $M' := M_{\ell+1}$ such that $M'(a) = M(a)$ for $a \in A \setminus A'$, $M'(a) = M_k(a) = N(a)$ for $a \in A_k \setminus \{b_k\}$ and $M'(a) = \emptyset_{b_k}$ for all $k \in [\ell]$, which completes the proof of the corollary.
\end{proof}

\restateLemBnoCycles*

\begin{proof}
        By contradiction assume that there is a directed cycle $C \subseteq B$. 
        Recall that~$B$ is the union of paths $P^M_a$ for $M \in \M$ and $a \in A^M_+$.
        We first observe that for any $M \in \M$ and $a \in A^M_+$, either $C \cap P^M_a = \emptyset$ or $C \cap P^M_a$ is a prefix of $P^M_a$, i.e., a subpath of $P^M_a$ starting at $a \in A_+$.
        This is true because no path $P^M_a$ can enter $C$ from outside, as every node has at most one incoming arc in $B$.
        Therefore, $C$ is the concatenation of paths, each corresponding to a prefix of $P^M_a$ for some $M \in \M$ and $a \in A^M_+$.
        Moreover, the cycle $C$ cannot contain any node of $A_-$, as these nodes do not have outgoing arcs in $B$.

        Let $\bar{D} = (A, \bar{E})$ with $\bar{E} := \bigcup_{M \in \M} E_{N \rightarrow M}$ be the digraph resulting from union of the exchange graphs.
        Note that $C \subseteq B \subseteq \bar{E}$.
        By the above analysis, $C$ is a directed cycle in $\bar{D}$ with the following properties:
        \begin{enumerate}
            \item No agent on $C$ is in $A_-$.\label{prop:cycle-no-Aminus}
            \item There are agents $a_1, \dots, a_{\ell} \in A_+$ on $C$ and paths $P_1, \dots, P_{\ell}$ contained in $C$ such that  the following holds:\label{prop:cycle-paths}
            \begin{enumerate}
                \item For each $i \in [\ell]$, $P_i$ is the unique $a_i$-$a_{i+1}$-path (where $a_{\ell+1} := a_1$) consistent with the direction of the cycle.\label{eq:cycle-paths-first}
                \item For each $i \in [\ell]$, there is $M_i \in \M$ such that $P_i$ is contained in $D_{M_i}$.
                \item For each $i \in [\ell]$, it holds that $M_i \neq M_{i+1}$ (where $M_{i+1} = M_1$).\label{eq:cycle-paths-last}
            \end{enumerate}
        \end{enumerate}
        We will iteratively replace $C$ by a smaller cycle in $\bar{D}$, maintaining properties \eqref{prop:cycle-no-Aminus} and~\eqref{prop:cycle-paths}, until we obtain the following additional property: 
        \begin{itemize}
            \item[(3)] For any $i, j \in [\ell]$ with $i > j$ and $M_i = M_j$, there is no arc $(a, b) \in E_{N \rightarrow M_i}$ with $a$ on $P_i$ and $b$ on $P_j$.
        \end{itemize}
        
        Indeed, note that if (3) is violated, we can replace $C$ by the concatenation of the paths $P_{j+1}, \dots, P_{i-1}$ and the path $P'$ resulting from traversing $P_i$ from the end of $P_{i-1}$ up to $a$, then taking $(a, b)$, and following $P_j$ from $b$ to the start of $P_{j+1}$.
        Note that the resulting cycle fulfills property~\eqref{prop:cycle-no-Aminus} because the agents on it are a subset of the original cycle $C$, and it fulfills property~\eqref{prop:cycle-paths} because $P'$ is contained in $E_{N \rightarrow M_i}$ and thus together with the paths $P_{j+1}, \dots, P_{i-1}$ and their respective endpoints, it fulfills properties \eqref{eq:cycle-paths-first} to \eqref{eq:cycle-paths-last} (after appropriate relabeling).
        Because this procedure reduces the number of arcs on the cycle, it can only be applied a finite number of times before we obtain a cycle $C$ meeting properties \eqref{prop:cycle-no-Aminus}, \eqref{prop:cycle-paths}, and $(3)$.

        Now for each $M \in \M$, let $I_M := \{i \in \ell : M_i = M\}$.
        If $I_M \neq \emptyset$ we can apply \cref{cor:exchange-path-multiple} to $M$ and the paths $P_i$ with $i \in I_M$, which by property~(3) fulfill the premise of the corollary, a set $A'_M$ of agents appearing on the paths $P_i$ for $i \in I_M$, with $a_i \in A'_M$ for each $i \in I_M$, and an {\Iconstr} matching $M'_M$ such that $M'_M(a) = N(a)$ for $a \in A'_M \setminus \{a_{i+1} : i \in I_M\}$. 
        Replacing each such $M$ with the corresponding $M'_M$ we obtain a family of {\Iconstr} matchings $\M'$ with $|\M'| = |\M|$.
        Now consider any agent $a \in A$.
        Note that if $a \in A'_M \notin \{a_1, \dots, a_{\ell}\}$ for some $\bar{M} \in \M$, then $M'_{\bar{M}}(a) = N(a) \not\prec_a M(a)$ for all $M \in \M$.
        Moreover, if $a = a_i$ for some $i \in [\ell]$, then $M'_{M_i}(a) = N(a) \succ_a M(a)$ for all $M \in \M$.
        Finally, for all other agents, the objects assigned by the matchings in $\M$ are exactly the same as those assigned in $\M$.
        Thus $\M'$ dominates $\M$, contradicting the latter's Pareto-optimality.
    \end{proof}

\section{Missing Proofs from \cref{sec:partial-order-hardness}}
\label{app:hardness}

\subsection{Proof of \cref{thm:NP-condorcet}}
\label{app:NP-condorcet}

\restateThmNPCondorcet*

We split the proof into four parts. First we show containment in NP, then we describe the construction of our reduction (from \textsc{$\ell$-Dimensional Matching}), then we show how a Pareto-optimal matching in the resulting instance can be constructed from such an $\ell$-dimensional matching, and vice versa.
Together, these four steps comprise the entire proof of the theorem.

\subsubsection*{Containment in NP}
Note that for a given set of matchings $\mathcal{M}$ in $G = (A \cup O, E)$ we can check whether it is popular as follows: For $e = (a, o) \in E$, let $$w_e := \begin{cases}
    2 & \text{ if } o \succ_a M(a) \text{ for all } M \in \mathcal{M},\\
    0 & \text{ if } M(a) \succ_a o \text{ for some } M \in \mathcal{M},\\
    1 & \text{otherwise}.
\end{cases}$$
Note that this definition of $w$ implies $\phi(\mathcal{M}, M') = |A| - \sum_{e \in M'} w_e$ for every matching~$M'$ in~$G$.
Hence, $\mathcal{M}$ is popular if and only if $\sum_{e \in M'} w_e \leq |A|$ for every matching~$M'$ in~$G$. 
The latter can be easily checked by computing a maximum weight matching~$M'$ in~$G$ with respect to $w$.

\subsubsection*{Construction: Reduction from \textsc{$\ell$-Dimensional Matching}}
    We prove the theorem via reduction from \textsc{$\ell$-Dimensional Matching}:
    Given $\ell$ disjoint sets $X_1, \dots, X_{\ell}$ with $|X_i| = m$ for all $i \in [\ell]$ and $\mathcal{E} \subseteq X_1 \times \dots \times X_{\ell}$, find an \emph{$\ell$-dimensional matching}, i.e., a subset $\mathcal{E}' \subseteq \mathcal{E}$ such that every element of $X := X_1 \cup \dots \cup X_{\ell}$ appears in exactly one element of $\mathcal{E}'$. 
    It is well-known that \textsc{$\ell$-Dimensional Matching} is NP-hard for any fixed $\ell \geq 3$.

    To show NP-hardness for deciding the existence of a popular set of matchings of cardinality $k$, we reduce from \textsc{$\ell$-Dimensional Matching} with $\ell = k + 1$.
    For this, we first fix an ordering of the elements in each $X_j$, so that we can write $X_j := \{x_{j1}, \dots, x_{jm}\}$ for each $j \in [\ell]$.
    
    The set of agents $A$ is constructed as follows:
    \begin{itemize}
        \item For every $x \in X_{\ell}$, there is an agent $a_x$.
        \item For each $h \in [m]$, there is a set $\hat{A}_h$ of $(k-1)$ agents.
        \item Finally, there is a set $A_0$ of $k^3m + k^2m + 1$ agents.
    \end{itemize}
    Thus $A = A_0 \cup \{a_x : x \in X_{\ell}\} \cup \bigcup_{h \in [m]} \hat{A}_h$. 
    The agents within the same set $A_0$ or $\hat{A}_h$ for $h \in [m]$, respectively, will have identical preferences and are thus not further distinguished.
    
    The set of objects $O$ is constructed as follows:
    \begin{itemize}
        \item There is an object $o_{y}$ for every $y \in Y := X_1 \cup \dots \cup X_{\ell-1}$.
        \item There is an object $o_{x_1, \dots, x_{\ell}}$ for every $(x_1, \dots, x_{\ell}) \in \bar{\mathcal{E}} := (X_1 \times \dots \times X_{\ell}) \setminus \mathcal{E}$.
        \item There is an object $\bar{o}_i$ for every $i \in [\ell-1]$.
        \item For $j \in [\ell-1]$, and $h \in [m]$, there is an object $\hat{o}_{jh}$.
        \item Moreover, $O$ includes two additional sets of objects, $\tilde{O}$ with $|\tilde{O}| = k^2m$ and $O_0$ with $|O_0| = |A_0|$ (the objects within each respective set will be treated identically by the agents and are thus not further distinguished). 
    \end{itemize}

    The edges and preferences of the agents are as follows (where no preference is specified for a pair of objects adjacent to an agent, the agent is indifferent for these objects):
    \begin{itemize}
        \item Each agent $a_x$ for $x \in X_{\ell}$ is adjacent to each object $o_y$ for $y \in Y$, 
        to each object $\bar{o}_i$ for $i \in [\ell-1]$, and
        to each object $o_{x_1,\dots, x_{\ell}}$ for $(x_1, \dots, x_{\ell}) \in \bar{\mathcal{E}}$ with $x_{\ell} = x$. 
        For $i \in [\ell-1]$, the agent prefers any object $o_y$ with $y \in X_i$ to $\bar{o}_i$.
        Moreover, for any $(x_1, \dots, x_{\ell}) \in \bar{\mathcal{E}}$ with $x_{\ell} = x$ and any $i \in [\ell-1]$, the agent prefers any object $o_{y}$ with $y \in X_i \setminus \{x_i\}$ to $o_{x_1, \dots, x_{\ell}}$.
        \item Each agent $a \in \hat{A}_h$ for $h \in [m]$ is adjacent to the objects $o_{x_{jh}}$ and $\hat{o}_{jh}$ for each $j \in [\ell - 1]$.
        The agent prefers $o_{x_{jh}}$ over $\hat{o}_{jh}$ for $j \in [\ell-1]$.
        \item Every agent $a \in A_0$ is adjacent to each object $o_y$ for $y \in Y$,
        to each object in $O_0$, and to each object in $\tilde{O}$.
        The agent prefers any object $o_y$ for $y \in Y$ to any object in $O_0$.
    \end{itemize}
    In the following, we show that there is a popular set of $k$ matchings in the constructed instance if and only if there exists an $\ell$-dimensional matching.
    For the remainder of the analysis, we define $O_Y := \{o_y : y \in Y\}$.

\subsubsection*{Constructing a popular set of matchings from an $\ell$-dimensional matching}
    We show how to construct a popular set of $k$ matchings $M_1, \dots, M_k$ from an $\ell$-dimensional matching $\mathcal{E}' \subseteq \mathcal{E}$.
    To this end, we first construct a $k$-matching $M^*$.
    For $x \in X_{\ell}$, match agent $a_x$ to the objects $o_{x_1}, \dots, o_{x_{\ell - 1}}$ where $(x_1, \dots, x_{\ell})$ is the unique tuple in $\mathcal{E}'$ with $x_{\ell} = x$ (recall that $\ell = k+1$ and hence $a_x$ is matched to exactly $k$ objects).
    For $h \in [m]$ we match each agent in $\hat{A}_h$ to each object $o_{x_{jh}}$ for $j \in [\ell - 1]$ (recall that $|\hat{A}_h| = k - 1$ and note that $o_{x_{jh}}$ is matched to exactly one agent $a_x$ for some $x \in X_{\ell}$, and that therefore $o_{x_{jh}}$ is matched to $k$ agents in total).
    Moreover, $M^*$ contains an arbitrary perfect matching between the agents in $A_0$ and the objects in $O_0$.
    We show that any decomposition of $M^*$ into $k$ matchings $M_1, \dots, M_k$ is popular.
    For this, consider any matching $M'$ in $G$ and let
    \begin{align*}
        A^+ & := \{a \in A : M'(a) \succ_a M_i(a) \text{ for all } i \in [k]\}\text{ and }\\
        A^- & := \{a \in A : M_i(a) \succ_a M'(a) \text{ for some } i \in [k]\}.
    \end{align*}
    We will show that $|A^+| \leq |A^-|$, revealing that $M_1, \dots, M_k$ is popular.
    We first establish the following two claims.

    \begin{claim}\label{clm:NP-Aplus}
        If $a \in A^+$ then $a \in A_0$ and $M'(a) \in O_Y$.
    \end{claim}
    \begin{proof}\renewcommand{\qedsymbol}{$\blacklozenge$}
        If $a \in A^+$ then $M'(a) \succ_a M_i(a)$ for all $i \in [k]$.
        Note however, that every agent $a_x$ for $x \in X_{\ell}$ and every agent $a \in \hat{A}_h$ for $h \in [m]$ is assigned an undominated object $o \in O_Y$ in each $M_i$, and thus the only agents in $A^+$ must be from $A_0$. Recall that all agents in $A_0$ are matched to some object in $O_0$ in some $M_i$. Moreover, the only objects that agents in $A_0$ prefer over objects in $O_0$ are the objects in $O_Y$.
        Therefore, $a \in A^+$ implies $a \in A_0$ and $M'(a) \in O_Y$.
    \end{proof}
    
    \begin{claim}\label{clm:NP-ax-dominance}
        Let $a \in A \setminus A_0$. If $a \notin A^-$ then $M'(a) \in O_Y$.
    \end{claim}
    \begin{proof}\renewcommand{\qedsymbol}{$\blacklozenge$}
        We distinguish two cases. 
        
        First, assume $a = a_x$ for some $x \in X_{\ell}$.
        Recall that in the $k$-matching $M^*$, the agent $a$ is matched to the objects $o_{x_1}, \dots, o_{x_{\ell - 1}}$ where $(x_1, \dots, x_{\ell})$ is the unique tuple in $\mathcal{E}'$ with $x_{\ell} = x$.
        In particular, this implies that for each $j \in [\ell-1]$ there is $i \in [k]$ and $y \in X_j$ such that $M_i(a) = o_{y} \succ_a \bar{o}_j$.
        Moreover, consider any $(\bar{x}_1, \dots, \bar{x}_{\ell}) \in \bar{\mathcal{E}}$ with $\bar{x}_{\ell} = x$. 
        Because $(x_1, \dots, x_{\ell}) \in \mathcal{E}' \subseteq \mathcal{E}$ there must be $j \in [\ell-1]$ such that $\bar{x}_j \neq x_j$.
        But then there is $i \in [k]$ with $M_i(a) = o_{x_j} \succ_a o_{\bar{x}_1, \dots, \bar{x}_{\ell}}$.
        Hence, for every object $o$ that is adjacent to $a$ and not in $O_Y$, there is some $i \in [k]$ such that $M_i(a) \succ_a o$ in this case, which implies $M'(a) \in O_Y$.

        Second, if $a \in \hat{A}_h$ for $h \in [m]$, then for every $j \in [\ell-1]$ there is $i \in [k]$ such that $M_i(a) = o_{x_{jh}} \succ_a \hat{o}_{jh}$.
        Since the objects $\hat{o}_{jh}$ are the only objects adjacent to $a$ outside of $O_Y$, we again conclude that $M'(a) \in O_Y$, unless $a \in A^-$.
    \end{proof}
    
    Let $A' := \{a \in A \setminus A_0 : M'(a) \in O_Y\}$.
    Note that \cref{clm:NP-Aplus} implies $|A^+| \leq |O_Y| - |A'|$, while \cref{clm:NP-ax-dominance} implies $|A \setminus A_0| \leq |A^-| + |A'|$.
    Putting this together, we obtain $$|A^+| \leq |O_Y| - |A'| \leq |O_Y| + |A^-| - |A \setminus A_0| = (\ell-1)m + |A^-| - km = |A^-|$$
    as desired (where the final identity follows from $\ell-1 = k$).
    This implies that $M_1, \dots, M_k$ is indeed popular.

\subsubsection*{Constructing an $\ell$-dimensional matching from a popular set of matchings}
    We now show the converse how to construct an $\ell$-dimensional matching $\mathcal{E}' \subseteq \mathcal{E}$ from a popular set of $k$ matchings $M_1, \dots, M_k$ from.
    We start by establishing the following claim:
    \begin{claim}\label{clm:NP-domination}
            For every $a \in A \setminus A_0$, and every $o \in O \setminus O_Y$ that is adjacent to $a$, there is $i \in [k]$ such that $M_i(a) \succ_a o$.
    \end{claim}
    \begin{proof}\renewcommand{\qedsymbol}{$\blacklozenge$}
        By contradiction assume there is $a' \in A \setminus A_0$ and $o' \in O \setminus O_Y$ adjacent to $a'$ such that $M_i(a') \not\succ_a o'$ for all $i \in [k]$.
        We construct $M'$ with $\phi(\mathcal{M}, M') < 0$, contradicting the popularity of $\mathcal{M} := \{M_1, \dots, M_k\}$.

        Note that because $|A_0| = k^3m + k^2m + 1 > k(|\tilde{O}| + |O_Y|)$, there is at least one $a_0 \in A_0$ such that $M_i(a) \in O_0 \cup \{\emptyset\}$ for all $i \in [k]$.
        We start by including an arbitrary perfect matching between $A \setminus (A_0 \cup \{a'\}) \cup a_0$ and $O_Y$ in $M'$.
        Moreover, we match $a'$ to $o'$.
        We match all agents $a \in A_0$ with $M_i(a) \in O_Y$ for some $i \in [k]$ to some object in $\tilde{O}$ (note that this is possible because there are at most $k \cdot |O_Y| = k^2m \geq |\tilde{O}|$ such agents) and we match all other agents from $A_0$ arbitrarily to objects in $O_0$.
        
        Note that by construction of $M'$, there is no agent $a \in A$ with $M_i(a) \succ_a M'(a)$ for some $i \in [k]$ (in particular, $M'$ matches every agent in $a \in A \setminus (A_0 \cup \{a'\})$ to an  object that is maximimal w.r.t.~$\succ_a$ and it matches $a'$ to an object that is not dominated by $M_i(a')$ for any $i \in [k]$). 
        However, $M'(a_0) \succ_{a_0} M_i(a_0)$ for all $i \in [k]$ and therefore $\phi(\mathcal{M}, M') \leq -1$, completing the contradiction.
    \end{proof}
    From \cref{clm:NP-domination}, we conclude that for every $x \in X_{\ell}$ and every $j \in [\ell-1]$, there must be $i \in [k]$ such that $M_i(a_x) \succ_{a_x} \bar{o}_j$, which is only possible if $M_i(a_x) = o_{y}$ for some $y \in X_j$.
    In fact, because $\ell - 1 = k$, for each $j \in [\ell-1]$ there must be exactly one $i \in [k]$ such that $M_i(a_x) = o_{y}$ for some $y \in X_j$.
    Moreover, for every $j \in [\ell-1]$ and every $h \in [m]$ each of the $k - 1$ agents in $\bar{A}_h$ is matched to $o_{x_{jh}}$ in some $M_i$.
    With the above, this implies that for every $j \in [\ell-1]$ and every $h \in [m]$, there is exactly one $i \in [k]$ and exactly one $x \in X_{\ell}$ such that $M_i(a_x) = o_{x_{jh}}$.
    For each $x \in X_{\ell}$, let $E_x = (x_1, \dots, x_{\ell})$ with $x_{\ell} = x$ denote the unique tuple such that $x_j$ for $j \in [\ell-1]$ is the unique element of $X_j$ with $M_i(a_x) = x_j$ for some $i \in [k]$.
    We show that $E_x \in \mathcal{E}$.
    Assume by contradiction that $E_x \notin \mathcal{E}$, then $E_x \in \bar{\mathcal{E}}$ and hence there is $o_{x_1, \dots, x_{ell}} \in O \setminus O_Y$.
    Moreover, for every $i \in [k]$ it holds that $M_i(a_x) \not\succ_{a_x} o_{x_1, \dots, x_{ell}}$, because $M_i(a_x) \in \{x_1, \dots, x_{\ell-1}\}$.
    This contradicts \cref{clm:NP-domination}.

    Now let $\mathcal{E}' := \{E_x : x \in X_{\ell}\}$.
    Note that by construction every $x \in X_1 \cup \dots \cup X_{\ell}$ appears in exactly one tuple in $\mathcal{E}'$ and that $\mathcal{E}' \subseteq \mathcal{E}$, as shown above.
    Thus $\mathcal{E}'$ is an $\ell$-dimensional matching.
    With this, the proof of \cref{thm:NP-condorcet} is complete.

\subsection{Proof of \cref{thm:NP-pareto}}
\label{app:NP-Pareto}

Here, we provide the remaining argument for the proof of \cref{thm:NP-pareto}, showing that the matching instance constructed in \cref{sec:hardness-pareto} admits a Pareto-optimal matching if and only if there is a vertex cover of size $\ell$ in the original graph.

\subsubsection*{Constructing a Pareto-optimal matching from a vertex cover}
    We show how to construct a Pareto-optimal matching $M$ in $G$ from a vertex cover $U \subseteq V'$ in $G'$ with $|U| = \ell$ (note that assuming equality here is without loss of generality).
    \begin{itemize}
        \item For each $u \in U$, we match $a_u$ to $o_u$.
        \item For edge $e = \{v, w\}$, assuming without loss of generality that $v \in U$ (which is possible because $U$ is a vertex cover), we match $a_e$ to $o^w_e$, $a^v$ to $o^v_e$, $a^w_e$ to $o^1_e$, $a^1_e$ to $o^2_e$, $a^2_e$ to $o^2_e$, and leave $o^3_e$ unmatched.
        \item Moreover, we add an arbitrary perfect matching between the agents in $\{a_v : v \in V' \setminus U\}$ and the objects in $\bar{O}$ (note that both sets have cardinality $|V| - \ell$) to $M$.
    \end{itemize}
    Now assume by contradiction that there exists a matching $M'$ that Pareto-dominates $M$.
    We first observe that $M'(a_v) = M(a_v)$ for all $v \in V$.
    Indeed, $M'(a_v) \in \bar{O}$ for each $v \in V \setminus U$, as otherwise $M(a_v) \succ_{a_v} M'(a_v)$.
    But then $M'(a_u) \notin \bar{O}$ for $u \in U$, implying $M'(a_u) = o_u = M(a_u)$, again as otherwise $M(a_u) \succ_{a_u} M'(a_u)$.
    Now consider any $e = \{v, w\} \in E$, assuming w.l.o.g.\ that $v \in U$ and $M(a_e) = o^w_e$ and $M(a^v_e) = o^v_e$.
    Because $M'(a_e) = o_v$, agent $a_e$ cannot be matched to $o_v$ in $M'$.
    Because $M(a_e) = o^w_e$, agent $a_e$ cannot be matched to $o_w$ in $M'$ (otherwise $M(a_e) \succ_{a_e} M'(a_e)$).
    Because $M(a^v_e) = o^v_e$ is the unique top-choice object of $a^v_e$, it must also hold that $M'(a^v_e) = o^v_e$.
    Hence $M'(a_e) = o^w_e = M(a_e)$ and $M'(a^w_e) = o^1_e$ by exclusion of all other options.
    This further implies $M'(a^1_e) = M(a^1_e) = o^2_e$, as any other remaining option would be worse for agent $a^1_e$, and then, by the same argument, $M'(a^2_e) = M(a^2_e) = o^3_e$ and finally $M'(a^3_e) = M(a^3_e) = \emptyset$.
    Thus $M'$ equals $M$, except for possible permutations in the perfect matching between $\{a_v : v \in V \setminus U\}$ and $\bar{O}$.
    In particular, there is no agent $a \in A$ with $M'(a) \succ_a M(a)$, contradicting our assumption that $M'$ Pareto-dominates~$M$.
    We have thus shown that existence of a vertex cover of size at most $\ell$ in $G'$ implies existence of a Pareto-optimal matching in $G$.

\subsubsection*{Constructing a vertex cover from a Pareto-optimal matching}
    We show how to construct a vertex cover of size at most $\ell$ from a Pareto-optimal matching $M$ in $G$.
    We start by showing the following claim.
    
    \begin{claim}\label{clm:NP-o1}
        Let $e = \{v, w\} \in E$. Then $M(a^v_e) = o^1_e$ or $M(a^w_e) = o^1_e$.
    \end{claim}
    \begin{proof}\renewcommand{\qedsymbol}{$\blacklozenge$}
        Assume by contradiction there is $e = \{v, w\} \in E$ with $M(a^v_e) \neq o^1_e$ and $M(a^w_e) \neq o^1_e$.
        We distinguish two cases:
        
        First, assume there is $i \in \{1, 2, 3\}$ such that $M(a^i_e) = \emptyset$.
        Then one of the objects $o^j_e$ for $j \in \{1, 2, 3\}$ must be unassigned (as neither $a^v_e$ nor $a^w_e$ is matched to $o^1_e$.
        But this contradicts the Pareto-optimality of $M$, as we can improve $M$ simply by matching $a^i_e$ to $o^j_e$.

        Second, assume each agent $a^1_e, a^2_e, a^3_e$ must be assigned an object by $M$. By symmetry, we can assume without loss of generality that $M(a^i_e) = o^i_e$ for $i \in \{1, 2, 3\}$.
        But then the matching $M'$ with $M(a^1_e) = o^3_e$, $M(a^2_e) = o^1_e$, $M(a^3_e) = o^2_e$, and $M'(a) = M(a)$ for all other agent Pareto-dominates $M$, because $a^1_e$ and $a^3_e$ are indifferent between their objects in $M$ and $M'$, and $M'(a^2_e) \succ_a M(a^2_e)$. Thus, in both cases, we obtained a contradiction.
    \end{proof}

    We use \cref{clm:NP-o1} to show the following.
    
    \begin{claim}\label{clm:NP-vertex-cover}
        Let $e = \{v, w\} \in E$. Then $M(a_v) = o_v$ or $M(a_w) = o_w$.
    \end{claim}
    \begin{proof}\renewcommand{\qedsymbol}{$\blacklozenge$}
        Assume by contradiction there is $e = \{v, w\} \in E$ with $M(a_v) \neq o_v$ and $M(a_w) \neq o_w$.
        By \cref{clm:NP-o1}, we have $M(a^v_e) = o^1_e$ or $M(a^w_e) = o^1_e$ (by symmetry, we can assume the latter without loss of generality).
        But then $M(a_e) = o^w_e$, because otherwise matching $a^w_e$ to $o^w_e$ instead of $o^1_e$ improves $M$, contradicting its Pareto-optimality.
        However, as $M(a_v) \neq o_v$ by our initial assumption, $o_v$ is unmatched in $M$. 
        Hence, we can improve $M$ by matching $a_e$ to $o_v$ instead of $o^w_e$ (note that $a_e$ is indifferent between the two) and matching $a^w_e$ to $o^w_e$ instead of $o^1_e$, which the agent strictly prefers.
        We thus obtained a contradiction to Pareto-optimality.
    \end{proof}

    Note that \cref{clm:NP-vertex-cover} implies that $U := \{v \in V : M(a_v) = o_v\}$ is a vertex cover.
    To see that $|U| \leq \ell$, note that if $|U| > \ell$ there must be at least one object $o \in \bar{O}$ that is not matched to any agent in $M$. 
    But then we can improve $M$ by matching $a_u$ for some arbitrary $u \in U$ to $o$ instead of $a_u$, contradicting Pareto-optimality.
    We have thus shown that existence of a Pareto-optimal matching in the constructed matching instance implies the existence of a vertex cover of size at most $\ell$ in $G$. With this, the proof of \cref{thm:NP-pareto} is complete.

    \section{Counter Example for Assignment Problem} \label{sec:appassignment}

   In the \emph{assignment problem}, we are given a bipartite graph $G=(A \cup O, E)$ and preference relations $\succ_a$ over adjacent nodes of $a$ for $a \in A$. The set of feasible alternatives $\mathcal{X}$ consists of all \emph{perfect matchings} (assignments) in $G$. 
   
   \Cref{fig:counterAssignment} shows that there exist Pareto-optimal sets of size $2$ that are not weakly Condorcet-winning. Let $\M$ be the set of the two blue assignments (left). Agents $a_1$ and $a_4$ receive their top choice, and $\M$ is the unique set of assignments with this property; hence $\M$ is Pareto-optimal. However, the red assignment $N$ (right) is preferred by every agent except $a_1$ and $a_4$, so $\M$ is not weakly Condorcet winning. This construction extends to larger Pareto-optimal sets by adding more short paths (as in the upper and lower path) and lengthening the middle path.

    \begin{figure}[h]
        \centering
\includegraphics[width=\linewidth]{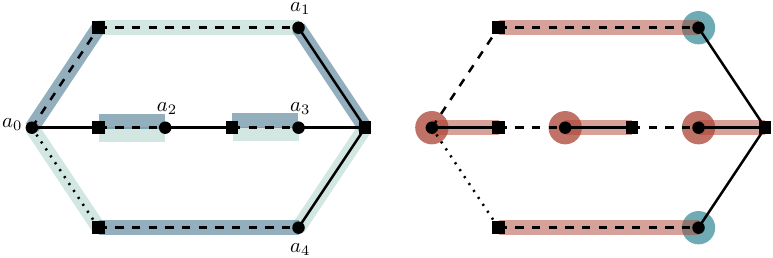}
        \caption{Assignment instance showing that Pareto optimal sets need not be Condorcet-winning in the assignment setting. Circles refer to agents and square to objects. The edge patterns indicate the agents' preferences over the objects, where solid is preferred over dashed over dotted. The left image illustrates the set of blue matchings $\M$ and the right image illustrates the competitor assignment $N$ in red. In the right image, agents are colored in accordance to their preferences when comparing $\M$ vs. $N$.}
        \label{fig:counterAssignment}
    \end{figure}

\end{document}